\documentclass{sig-alternate-2013}

\permission{Permission to make digital or hard copies of all or part of this work for personal or classroom use is granted without fee provided that copies are not made or distributed for profit or commercial advantage and that copies bear this notice and the full citation on the first page. Copyrights for components of this work owned by others than ACM must be honored. Abstracting with credit is permitted. To copy otherwise, or republish, to post on servers or to redistribute to lists, requires prior specific permission and/or a fee. Request permissions from permissions@acm.org.}
\conferenceinfo{SIGMOD'14,}{June 22--27, 2014, Snowbird, UT, USA.}
\copyrightetc{Copyright 2014 ACM \the\acmcopyr}
\crdata{978-1-4503-2376-5/14/06\ ...\$15.00.\\
http://dx.doi.org/10.1145/2588555.2588581}

\usepackage{graphicx,epsfig,subfigure}
\usepackage{amsfonts,amsmath,amssymb,color,latexsym,multirow,xspace,color,url}
\usepackage[boxed]{algorithm}
\usepackage[noend]{algorithmic}
\usepackage{hyperref}

\usepackage{balance}  

\newcommand{\pred}{{\phi}\xspace}
\newcommand{\one}{{\mathbf{1}}\xspace}
\newcommand{\query}{{\sf q}\xspace}
\newcommand{\maxcomp}{{\rm maxcomp}\xspace}

\renewcommand{\S}{\ensuremath{{\cal S}}}
\newcommand{\Spairs}{\ensuremath{{\cal S}_{\rm pairs}}} 
\newcommand{\inp}{\ensuremath{{\cal I}}}
\newcommand{\dom}{\ensuremath{{\cal T}}}

\newcommand{\dist}{d} 
\newcommand{\error}{\ensuremath{{\cal E}}}

\newcommand{\squishlist}{
   \begin{list}{$\bullet$}
    {
      \setlength{\itemsep}{0pt}
      \setlength{\parsep}{3pt}
      \setlength{\topsep}{3pt}
      \setlength{\partopsep}{0pt}
      \setlength{\leftmargin}{1.5em}
      \setlength{\labelwidth}{1em}
      \setlength{\labelsep}{0.5em} } }

\newcommand{\squishend}{
    \end{list}  }
\newcommand{\squishenum}{
   
   \begin{list}{scount}{ \usecounter{scount}}
    {
      \setlength{\itemsep}{0pt}
      \setlength{\parsep}{3pt}
      \setlength{\topsep}{3pt}
      \setlength{\partopsep}{0pt}
      \setlength{\leftmargin}{1.5em}
      \setlength{\labelwidth}{1em}
      \setlength{\labelsep}{0.5em} } }

\newtheorem{definition}{Definition}[section]

\newtheorem{example}{Example}[section]
\newtheorem{lemma}{Lemma}[section]

\newtheorem{theorem}[lemma]{Theorem}

\newtheorem{corollary}[lemma]{Corollary}

        {\hspace*{\fill}$\Box$\par\vspace{4mm}}
        {\hspace*{\fill}$\Box$\par}

\newcommand{\stitle}[1]{\smallskip \noindent{\bf #1}}

\newcommand{\nop}[1]{{}\xspace}
\newcommand{\hide}[1]{\hspace*{-5pt}\xspace}

\newcommand{\newstuff}[1]{#1}

\newcommand{\ie}{{i.e.}\xspace}

\newcommand{\csec}{Section~}

\newcommand{\cfig}{Figure~}

\newcommand{\cthm}{Theorem~}

\newcommand{\cexa}{Example~}

\newcommand{\xxx}[1]{\xspace}

\newcommand{\ceil}[1]{\lceil #1 \rceil}

\newcommand{\true}{{\rm true}\xspace}
\newcommand{\false}{{\rm false}\xspace}
\newcommand{\size}{{\rm size}\xspace}

\newcommand{\lat}{{\cal L}\xspace}

\def\blfootnote{\xdef\@thefnmark{}\footnotetext} 

\begin{document}
\conferenceinfo{SIGMOD}{'14 Snowbird, Utah USA}


\title{Blowfish Privacy: Tuning Privacy-Utility Trade-offs using Policies}


\numberofauthors{3}

\author{
\alignauthor
Xi He\\
	\affaddr{Duke University}\\
	\affaddr{Durham, NC, USA}\\
	\email{hexi88@cs.duke.edu}
\alignauthor
Ashwin Machanavajjhala\\
       \affaddr{Duke University}\\
       \affaddr{Durham, NC, USA}\\
       \email{ashwin@cs.duke.edu}
\alignauthor
Bolin Ding\\
       \affaddr{Microsoft Research}\\
       \affaddr{Redmond, WA, USA}\\
       \email{bolin.ding@microsoft.com}
}

\maketitle
\begin{abstract}
Privacy definitions provide ways for trading-off the privacy of individuals in a statistical database for the utility of downstream analysis of the data. In this paper, we present {\em Blowfish}, a class of privacy definitions inspired by the Pufferfish framework, that provides a rich interface for this trade-off. In particular, we allow data publishers to extend differential privacy using a {\em policy}, which specifies (a) {\em secrets}, or information that must be kept secret, and (b) {\em constraints} that may be known about the data. While the secret specification allows increased utility by lessening protection for certain individual properties, the constraint specification provides added protection against an adversary who knows correlations in the data (arising from constraints). We formalize policies and present novel algorithms that can handle general specifications of sensitive information and certain count constraints. We show that there are reasonable policies under which our privacy mechanisms for k-means clustering, histograms and range queries  introduce significantly lesser noise than their differentially private counterparts. We quantify the privacy-utility trade-offs for various policies analytically and  empirically on real datasets.
\end{abstract}

\vspace{-3mm}
\vspace{-3mm}

\section{Introduction}
\label{sec:intro}
With the increasing popularity of ``big-data'' applications which collect, analyze and disseminate individual level information in literally every aspect of our life, ensuring that these applications do not breach the privacy of individuals is an important problem. The last decade has seen the development of a number of privacy definitions and mechanisms that trade-off the privacy of individuals in these databases for the utility (or accuracy) of data analysis (see \cite{chen09:fnt} for a survey). Differential privacy \cite{icalp:Dwork06} has emerged as a gold standard not only because it is not susceptible to attacks that other definition can't tolerate, but also since it provides a simple knob, namely $\epsilon$, for trading off privacy for utility. 

While $\epsilon$ is intuitive, it does not sufficiently capture the diversity in the privacy-utility trade-off space. For instance, recent work has shown two seemingly contradictory results. In certain applications (e.g., social recommendations \cite{ashwin11:vldb}) differential privacy is too strong and does not permit sufficient utility. Next, when data are correlated (e.g., when constraints are known publicly about the data, or in social network data) differentially private mechanisms may not limit the ability of an attacker to learn sensitive information \cite{sigmod:KiferM11}. 
Subsequently, Kifer and Machanavajjhala \cite{pods:KiferM12} proposed a semantic privacy framework, called Pufferfish, which helps clarify assumptions underlying privacy definitions -- specifically, the information that is being kept secret, and the adversary's background knowledge. They showed that differential privacy is equivalent to a specific instantiation of the Pufferfish framework, where (a) every property about an individual's record in the data is kept secret, and (b) the adversary assumes that every individual is independent of the rest of the individuals in the data (no correlations).
We believe that these shortcomings severely limit the applicability of differential privacy to real world scenarios that either require high utility, or deal with correlated data. 

Inspired by Pufferfish, we seek to  better explore the trade-off between privacy and utility by providing a richer set of ``tuning knobs''. We explore a class of definitions called {\em Blowfish privacy}. In addition to $\epsilon$, which controls the amount of information disclosed, Blowfish definitions take as input a privacy {\em policy} that specifies two more parameters -- which information must be kept secret about individuals, and what constraints may be known publicly about the data. By extending differential privacy using these policies, we can hope to develop mechanisms that permit more utility since not all properties of an individual need to be kept secret. Moreover, we also can limit adversarial attacks that leverage correlations due to publicly known constraints. 

\noindent We make the following contributions in this paper:
\squishlist
\item We introduce and formalize sensitive information specifications, constraints, policies and Blowfish privacy. We consider a number of realistic examples of sensitive information specification, and focus on count constraints.
\item We show how to adapt well known differential privacy mechanisms to satisfy Blowfish privacy, and using the example of k-means clustering illustrate the gains in accuracy for Blowfish policies having weaker sensitive information specifications.
\item We propose  the ordered mechanism, a novel strategy for releasing cumulative histograms and answering range queries. We show analytically and using experiments on real data that, for reasonable sensitive information specifications, the ordered hierarchical mechanism is more accurate than the best known differentially private mechanisms for these workloads. 
\item We study how to calibrate noise for policies expressing count constraints, and its applications in several practical scenarios.
\squishend
{\bf Organization:} Section~\ref{sec:notation} introduces the notation. Section~\ref{sec:policy} formalizes privacy policies. We define Blowfish privacy, and discuss composition properties and its relationship to prior work in Section~\ref{sec:blowfish}. 
We define the policy specific global sensitivity of queries in Section~\ref{sec:noconstraints}. We describe mechanisms for kmeans clustering (Section~\ref{sec:kmeans}), and releasing cumulative histograms \& answering range queries (Section~\ref{sec:cdf}) under  Blowfish policies without constraints and empirically evaluate the resulting privacy-utility trade-offs on real datasets.
We show how to release histograms in the presence of count constraints in Section~\ref{sec:auxiliary} and then conclude in Section~\ref{sec:conclusions}.

\section{Notation}
\label{sec:notation}
We consider a dataset $D$ consisting of $n$ tuples. Each tuple $t$ is considered to be drawn from a
domain $\dom = A_1 \times A_2 \times \ldots \times A_m$ constructed from the cross product of $m$
categorical attributes. We assume that each tuple $t$ corresponds to the data collected from a unique
individual with identifier $t.\_id$. We will use the notation $x\in \dom$ to denote a value in the
domain, and $x.A_i$ to denote the $i^{th}$ attribute value in $x$.  

Throughout this paper, we will make an assumption that the set of individuals in the dataset $D$ is known in advance to the adversary and does not change. Hence we will use the indistinguishability notion of differential privacy \cite{tcc:DworkMNS06}. We will denote the set of possible databases using $\inp_n$, or the set of databases with $|D| = n$.\footnote{\small In Sec.~\ref{sec:policy} we briefly discuss how to generalize our results to other differential privacy notions by relaxing this assumption.}

\xxx{
\begin{definition}[Neighbors]
Two datasets $D_1$ and $D_2$ are neighbors, denoted by $(D_1, D_2) \in N)$, if they differ in the
value of one tuple. 
\end{definition}
}
 
\begin{definition}[Differential Privacy \cite{icalp:Dwork06}]\label{def:dp}
Two \\ datasets $D_1$ and $D_2$ are neighbors, denoted by $(D_1, D_2) \in N$, if they differ in the
value of one tuple. A randomized mechanism $M$ satisfies $\epsilon$-differential privacy if for every set of outputs $S
\subseteq range(M)$, and every pair of neighboring datasets $(D_1$, $D_2) \in N$, 
\begin{equation}
Pr[M(D_1) \in S] \leq e^\epsilon Pr[M(D_2) \in S]
\end{equation}  
\end{definition} 

Many techniques that satisfy differential privacy use the following notion of {\em global sensitivity}: 
\begin{definition}[Global Sensitivity]
The global sensitivity of a function $f:\inp_n \rightarrow \mathbb{R}^d$, denoted by $S(f)$ is defined as the largest L1 difference $||f(D_1) - f(D_2)||_1$, where $D_1$ and $D_2$ are databases that differ in one tuple. More formally, 
\begin{equation}
S(f) \ = \ \max_{(D_1, D_2) \in N} ||f(D_1) - f(D_2)||_1
\end{equation}
\end{definition}

A popular technique that satisfies $\epsilon$-differential privacy is the Laplace mechanism \cite{tcc:DworkMNS06} defined as follows:
\begin{definition}\label{def:laplace}
The Laplace mechanism, $M^{Lap}$, privately computes a function $f: \inp_n \rightarrow \mathbb{R}^d$ by computing $f(D) + \mathbf{\eta}$. $\mathbf{\eta} \in \mathbb{R}^d$ is a vector of independent random variables, where each $\eta_i$ is drawn from the Laplace distribution with parameter $S(f)/\epsilon$. That is, $P[\eta_i = z] \propto e^{-z\cdot \epsilon / S(f)}$.
\end{definition} 

Given some partitioning of the domain ${\cal P} = (P_1, \ldots, P_k)$, we denote by $h_{\cal P}: \inp \rightarrow Z^{k}$ the histogram query. $h_{\cal P}(D)$ outputs for each $P_i$ the number of times values in $P_i$ appears in $D$. $h_\dom(\cdot)$ (or $h(\cdot)$ in short) is the {\em complete} histogram query that reports for each $x \in \dom$ the number of times it appears in $D$. It is easy to see that $S(h_{\cal P}) = 2$ for all histogram queries, and the Laplace mechanism adds noise proportional to $Lap(2/\epsilon)$ to each component of the histogram.
We will use Mean Squared Error as a measure of accuracy/error. 
\begin{definition}\label{def:mse}
Let $M$ be a randomized algorithm that privately computes a function $f:\inp_n \rightarrow \mathbb{R}^d$. The expected mean squared error of $M$ is given by:
\begin{equation}
\error_M(D) \ = \ \sum_{i} \mathbb{E}(f_i(D) - \tilde{f_i}(D))^2
\end{equation}
where $f_i(\cdot)$ and $\tilde{f_i}(\cdot)$ denote the $i^{th}$ component of the true and noisy answers, respectively.
\end{definition}
Under this definition the accuracy of the Laplace mechanism for histograms is given by $|\dom|\cdot\mathbb{E}(Laplace(2/\epsilon))^2 = 8|\dom|/\epsilon^2$.

\section{Policy Driven Privacy}
\label{sec:policy}
In this section, we describe an abstraction called a {\em policy} that helps specify which information has to be kept secret and what background knowledge an attacker may possess about the correlations in the data. We will use this policy specification as input in our privacy definition, called Blowfish, described in Section~\ref{sec:blowfish}.

\subsection{Sensitive Information}
\label{sec:sensitive_example}
As indicated by the name, Blowfish\footnote{\small Pufferfish and Blowfish are common names of the same family of marine fish, Tetraodontidae.} privacy is inspired by the Pufferfish privacy framework \cite{pods:KiferM12}. In fact, we will show later (in Section~\ref{sec:related}) that Blowfish privacy is equivalent to specific instantiations of semantic definitions arising from the Pufferfish framework. 
 
Like Pufferfish, Blowfish privacy also uses the notions of secrets and discriminative pairs of secrets. We define a secret to be an arbitrary propositional statement over the values in the dataset. For instance, the secret $s:\, t.\_id = \mbox{`Bob'} \wedge t.Disease = \mbox{`Cancer'}$ is true in a dataset where Bob has Cancer.  We denote by $\S$ a set of secrets that the data publisher would like to protect. As we will see in this section each individual may have multiple secrets. Secrets may also pertain to sets of individuals. For instance, the following secret $s: t_1.\_id = \mbox{`Alice'} \wedge t_2.\_id = \mbox{`Bob'} \wedge t_1.Disease = t_2.Disease$ is true when Alice and Bob have the same disease. However, in this paper, we focus on the case where each secret is about a single individual.

We call a pair of secrets $(s, s') \in \S \times \S$ {\em discriminative} if
they are mutually exclusive. Each discriminative pair describes properties that an adversary must not
be able to distinguish between. One input to a policy is a set of discriminative pairs of
secrets $\Spairs$.

We now present  a few examples of sensitive information specified as a set of discriminative secrets. 

\squishlist
\item {\em Full Domain}: Let $s^i_x$ be the secret $(t.\_id = i \wedge t = x)$, for some $x \in \dom$. We define $\Spairs^{\rm full}$ as: 
\begin{equation}
\Spairs^{\rm full} \ = \ \{(s^i_x, s^i_y) | \forall i, \forall (x,y) \in \dom \times \dom\}
\end{equation}
This means that for every individual, an adversary should not be able to distinguish whether that individual's value is $x$ or $y$, for all $x, y \in \dom$. 
\item {\em Attributes}: 
Let  $\mathbf{x} \in \dom$ denote a multidimensional value. Let $\mathbf{x}[A]$ denote value of attribute $A$, and $\mathbf{x}[\bar{A}]$ the value for the other attributes. 
Then a second example of sensitive information is:
\begin{equation}
\Spairs^{\rm attr} \ = \ \{(s^{i}_\mathbf{x}, s^{i}_\mathbf{y}) | \forall i, \exists A, \mathbf{x}[A] \neq \mathbf{y}[A] 
\wedge \mathbf{x}[\bar{A}] = \mathbf{y}[\bar{A}]\}
\end{equation}
$\Spairs^{attr}$ ensures that an adversary should not be able to sufficiently distinguish between any two values for each attribute of every individual's value. 
\item {\em Partitioned}: Let ${\cal P} = \{P_1, \ldots, P_p\}$ be a partition that divides the domain into $p$ disjoint sets ($\cup_i P_i = \dom$ and $\forall 1
\leq i,j \leq p, P_i \cap P_j = \emptyset$). We define partitioned sensitive information as: 
\begin{equation}
\Spairs^{\cal P} \ = \ \{(s^i_x, s^i_y) | \forall i, \exists j, (x,y) \in P_j \times P_j\}
\end{equation}
In this case, an adversary is allowed to deduce whether  an
individual is in one of two different partitions, but can't distinguish between two values within a single
partition. This is a natural specification for location data -- an individual may be
OK with releasing his/her location at a coarse granularity (e.g., a coarse grid), but location within
each grid cell must be hidden from the adversary.
\item {\em Distance Threshold}: In many situations there is an inherent distance metric $\dist$ associated with the points in the domain (e.g., $L_1$ distance on age or salary, or Manhattan distance on locations). Rather than requiring that an adversary should not be able to distinguish  between any pairs of points $x$ and $y$, one could require that each pair of points that are close are not distinguishable. So, for this purpose, the set of discriminative secrets is:
\begin{equation}
\Spairs^{\dist, \theta} \ = \ \{(s^i_x, s^i_y) | \forall i, \dist(x,y) \leq \theta \}
\end{equation} 
Under this policy, the adversary will not be able to distinguish any pair of values with certainty. However, the adversary may distinguish points that are farther apart better that points that are close. 
\squishend

\noindent All of the above specifications of sensitive information can be generalized using the {\em discriminative secret graph}, defined below. Consider a graph $G = (V, E)$, where $V = \dom$ and the set of edges $E \subseteq \dom \times \dom$. The set of edges can be interpreted as values in the domain that an adversary must not distinguish between; i.e., the set of discriminative secrets is $\Spairs^{G} = \{(s^i_x, s^i_y) \ | \ \forall i, \forall (x,y) \in E\}$.
The above examples correspond to the following graphs: $G^{\rm full}$ corresponds to a complete graph on all the elements in $\dom$. $G^{\rm attr}$ corresponds to a graph where two values are connected by an edge when only one attribute value changes. $G^{\cal P}$ has $|{\cal P}|$ connected components, where each component is a complete graph on vertices in $P_i$. Finally, in $G^{d, \theta}$, $(x, y) \in E$ iff $d(x,y) \leq \theta$. 

We would like to note that a policy could have secrets and discriminative pairs about sets of individuals. However, throughout this paper, we only consider secrets pertaining to a single individual, and thus discriminative pairs refer to two secrets about the same individual. Additionally, the set of discriminative pairs is the same for all individuals. One can envision different individuals having different sets of discriminative pairs. For instance, we can model an individual  who is privacy agnostic and does not mind disclosing his/her value exactly by having no discriminative pair involving that individual. Finally note that in all of the discussion in this section, the specification of what is sensitive information {\em does not depend on the original database $D$}. One could specify sensitive information that depends on $D$, but one must be wary that this might leak additional information to an adversary. In this paper, we focus on data-independent discriminative pairs, uniform secrets and secrets that only pertain to single individuals. 

Throughout this paper, we will assume that the adversary knows the total number of tuples in the database (i.e., the set of possible instances is $\inp_n$). Hence, we can limit ourselves to considering changes in tuples (and not additions or deletions). We can in principle relax this assumption about cardinality, by adding an additional set of secrets of the form $s^i_\bot$ which mean {\em ``individual $i$ is not in dataset''}. All of our definitions and algorithms can be modified to handle this case by adding $\bot$ to the domain and to the discriminative secret graph $G$. We defer these extensions to future work.
\ \\
\subsection{Auxiliary Knowledge}
Recent work \cite{sigmod:KiferM11} showed that differentially private mechanisms could still lead to an inordinate disclosure of sensitive information when adversaries have access to publicly known constraints about the data that induce correlations across tuples. This can be illustrated by the following example. Consider a table $D$ with one attribute $R$ that takes values $r_1, \ldots, r_k$. Suppose, based on publicly released datasets the following $k-1$ constraints are already known: $c(r_1) + c(r_2)  = a_1$, $c(r_2) + c(r_3) = a_2$, and so on, where $c(r_i)$ is the number of records with value $r_i$. This does not provide  enough information to always reconstruct the counts in $D$ ($k$ unknowns but $k-1$ linear equations). However, if we knew the answer to some $c(r_i)$, then all counts can be reconstructed -- in this way tuples are correlated.

Differential privacy allows answering all the count queries $c(r_i)$ by adding independent noise with variance $2/\epsilon^2$ to each count. While these noisy counts $\tilde{c}(r_i)$ themselves do not disclose information about any individual, they can be combined with the constraints to get very precise estimates of $c(r_i)$. That is, we can construct $k$ independent estimators for each count as follow. For $r_1$, $\tilde{c}(r_1), a_1 -\tilde{c}(r_2), a_1 - a_2 + \tilde{c}(r_3), \ldots$ each equal $c(r_1)$ in expectation and have a variance of $2/\epsilon^2$. By averaging these estimators, we can predict the value of $c(r_i)$ with a variance of $2/(k\epsilon^2)$. For large $k$ (e.g., when there are $2^d$ values in $R$), the variance is small so that the table $D$ is reconstructed with very high probability, thus causing a complete breach of privacy.

Therefore, our policy specification also takes into account auxiliary knowledge that an adversary might know about the individuals in the private database. In Blowfish, we consider knowledge in the form of a set of deterministic constraints $Q$ that are publicly known about the dataset. We believe these are easier to specify than probabilistic correlation functions for data publishers. The effect of the constraints in $Q$ is to make only a subset of the possible database instances $\inp_{Q} \subset \inp_n$ possible; or equivalently, all instances in $\inp_n \setminus \inp_Q$ are impossible. For any database $D \in \inp_n$, we denote by $D \vdash Q$ if $D$ satisfies the constraints in $Q$; i.e., $D \in \inp_Q$. Examples of deterministic constraints include: 
\squishlist
\item {\em Count Query Constraints:} A count query on a database returns the number of tuples that satisfy a certain predicate. A count query constraints is a set of (count query, answer) pairs over the database that are publicly known. 
\item {\em Marginal Constraints:} A marginal is a projection of the database on a subset of attributes, and each row counts the number of tuples that agree on the subset of attributes. The auxiliary knowledge of marginals means these database marginals  are known to the adversary. 
\squishend
\subsection{Policy}
\vspace{-3mm}
\newstuff{
\begin{definition}[Policy]
A policy is a triple\\ $P = (\dom, G, \inp_Q)$, where $G = (V,E)$ is a discriminative secret graph with $V \subseteq \dom$. 
In $P$, the set of discriminative pairs $\Spairs^{G}$ is defined as the set $\{(s^i_x, s^i_y) \, | \, \forall i \in \_id, \forall (x,y) \in E\}$, where $s^i_x$ denotes the statement: $t.\_id = i \wedge t = x$. $\inp_Q$ denotes the set of databases that are possible under the constraints $Q$ that are known about the database.
\end{definition}
}
Note that the description of the policy can be exponential in the size of the input dataset. We will use shorthand to describe certain types of sensitive information (e.g., full domain, partition, etc), and specify the set of possible databases $\inp_Q$ using the description of $Q$.

\section{Blowfish Privacy}
\label{sec:blowfish}
In this section, we present our new privacy definition, called Blowfish Privacy. 
Like differential privacy, Blowfish uses the notion of neighboring datasets. The key difference is that the set of neighbors in Blowfish depend on the policy $P$ -- both on the set of discriminative pairs as well as on the constraints known about the database. 


\begin{definition}[Neighbors]
Let $P = (\dom, G, \inp_Q)$ be a policy. For any pair of datasets $D_1, D_2$, let $T(D_1, D_2) \subseteq \Spairs^{G}$ be the set of discriminative pairs $(s^i_x, s^i_y)$ such that the $i^{th}$ tuples in $D_1$ and $D_2$ are $x$ and $y$, resp. Let $\Delta(D_1, D_2) = D_1 \setminus D_2 \cup D_2 \setminus D_1$. $D_1$ and $D_2$ are neighbors with respect to a policy $P$, denoted by  $(D_1, D_2) \in N(P)$, if: 
\begin{enumerate}
\vspace{-2mm}
\item $D_1, D_2 \in \inp_Q$. (i.e., both the datasets satisfy $Q$). 
\vspace{-2mm}
\item $T \neq \emptyset$. (i.e., $\exists (s^i_x, s^i_y) \in \Spairs^{G}$ such that the $i^{th}$ tuples in $D_1$ and $D_2$ are $x$ and $y$, resp). 
\vspace{-2mm}
\item There is no database $D_3 \vdash Q$ such that 
\begin{enumerate}
\vspace{-2mm}
\item $T(D_1, D_3) \subset T(D_1, D_2)$, or 
\vspace{-1mm}
\item $T(D_1, D_3) = T(D_1, D_2)$ \& $\Delta(D_3, D_1) \subset \Delta(D_2, D_1) $.
\end{enumerate} 
\end{enumerate}
\end{definition} 
When $P = (\dom, G, \inp_n)$ (i.e., no constraints), $D_1$ and $D_2$ are neighbors if some individual tuples value is changed from $x$ to $y$, where $(x,y)$ is an edge in $G$. Note that $T(D_1, D_2)$ is non-empty and has the smallest size (of 1). Neighboring datasets in differential privacy correspond to neighbors when $G$ is a complete graph. 

For policies having constraints, conditions 1 and 2 ensure that neighbors satisfy the constraints (i.e., are in $\inp_Q$), and that they differ in at least one discriminative pair of secrets. Condition 3 ensures that $D_1$ and $D_2$ are minimally different in terms of discriminative pairs and tuple changes. 

\begin{definition}[Blowfish Privacy]
Let $\epsilon > 0$ be a real number and $P = (\dom, G, \inp_Q)$ be a policy. A randomized mechanism $M$ satisfies $(\epsilon, P)$-Blowfish privacy if for every pair of neighboring databases $(D_1, D_2) \in N(P)$, and every set of outputs $S \subseteq range(M)$, we have 
\begin{equation}
Pr[M(D_1) \in S] \ \leq e^{\epsilon} Pr[M(D_2) \in S]
\end{equation}
\end{definition}  
Note that Blowfish privacy takes in the policy $P$ in addition to $\epsilon$ as an input, and is different from differential privacy in only the set of neighboring databases $N(P)$. For $P = (\dom, G, \inp_n)$ (i.e., no constraints), it is easy to check that for any two databases that arbitrarily differ in one tuple ($D_1 = D \cup \{x\}, D_2 = D\cup\{y\}$), and any set of outputs $S$, 
\begin{equation}
Pr[M(D_1) \in S] \ \leq e^{\epsilon\cdot  d_G(x,y)} Pr[M(D_2) \in S]
\end{equation}
where $d_G(x,y)$ is the shortest distance between $x, y$ in $G$. This implies that an attacker may better distinguish pairs of points farther apart in the graph (e.g., values with many differing attributes in  $\Spairs^{\rm attr}$), than those that are closer.  Similarly, an attack can distinguish between $x,y$ with probability 1, when $x$ and $y$ appear in different partitions under partitioned sensitive information $\Spairs^{\cal P}$ ($d_G(x,y) = \infty$).  


\subsection{Composition}
Composition \cite{kdd:GantaKS08} is an important property that any privacy notion should satisfy in order to be able to reason about independent data releases. 
Sequential composition ensures that a sequence of computations that each ensure privacy in isolation also ensures privacy. This allows breaking down computations into smaller building blocks. Parallel composition is crucial to ensure that too much error is not introduced on computations occurring on {\em disjoint} subsets of data. We can show that Blowfish satisfies sequential composition, and a weak form of parallel composition. 

\begin{theorem}[Sequential Composition]
Let $P=(\dom, G, \inp_Q)$ be a policy and $D \in \inp_Q$ be an input database. Let $M_1(\cdot)$ and $M_2(\cdot, \cdot)$ be algorithms with independent sources of randomness that satisfy $(\epsilon_1, P)$ and $(\epsilon_2, P)$-Blowfish privacy, resp. Then an algorithm that outputs both $M_1(D) = \omega_1$ and $M_2(\omega_1, D) = \omega_2$ satisfies $(\epsilon_1 + \epsilon_2, P)$-Blowfish privacy.
\end{theorem}
\begin{proof} See Appendix~\ref{proof_seqComp}.
\end{proof}

\begin{theorem}
{\sc (Parallel Composition with Cardinality Constraint).}
Let $P=(\dom, G, \inp_n)$ be a policy where the cardinality of the input $D \in \inp_n$ is known. Let $S_1, \ldots, S_p$ be disjoint subsets of $\_id$s; $D \cap S_i$ denotes the dataset restricted to the individuals in $S_i$. Let $M_i$ be mechanisms that each ensure $(\epsilon_i, P)$-Blowfish privacy. Then the sequence of $M_i(D \cap S_i)$ ensures $(\max_i \epsilon_i, P)$-Blowfish privacy. 
\end{theorem}
\begin{proof} See Appendix~\ref{proof_parComp}.
\end{proof}

Reasoning about parallel composition in the presence of general constraints is non-trivial. Consider two neighboring datasets $D_a, D_b \in N(P)$. For instance, suppose one of the attributes is gender, we know the number of males and females in the dataset, and we are considering full domain sensitive information. Then there exist neighboring datasets such that differ in two tuples $i$ and $j$ that are alternately male and female in $D_a$ and $D_b$. If $i$ and $j$ appear in different subsets $S_1$ and $S_2$ resp., then $D_a \cap S_1 \neq D_b \cap S_2$ and $D_a \cap S_1 \neq D_b \cap S_2$. Thus the sequence $M_i(D \cap S_i)$ does not ensure $(\max_i \epsilon_i, P)$-Blowfish privacy. We generalize this observation below.

Define a pair of secrets $(s,s')$ to be {\em critical} to a constraint $q$ if there exist $D_s, D_{s'}$ such that $T(D_s, D_{s'}) = (s,s')$, and $D_s \vdash q$, but $D_{s'} \not\vdash q$. Let $crit(q)$ denote the set of secret pairs that are critical to $q$. Next, consider $S_1, \ldots, S_k$ disjoint subsets of ids. We denote by $SP(S_i)$ the set of secret pairs that pertain to the ids in $S_i$. We say that a constraint $q$ {\em affects} $D \cap S_i$ if $crit(q) \cap SP(S_i) \neq \emptyset$. We can now state a sufficient condition for parallel composition. 
\begin{theorem}
{\sc (Parallel Composition with General Constraints).}
Let $P=(\dom, G, \inp_Q)$ be a policy and $S_1, \ldots, S_p$ be disjoint subsets of $\_id$s. Let $M_i$ be mechanisms that each ensure $(\epsilon_i, P)$-Blowfish privacy. Then the sequence of $M_i(D \cap S_i)$ ensures $(\max_i \epsilon_i, P)$-Blowfish privacy if there exist {\em disjoint} subsets of constraints  $Q_1, \ldots, Q_p \subset Q$ such that all the constraints in $Q_i$ only affects $D \cap S_i$.
\end{theorem}
\begin{proof} See Appendix~\ref{proof_parComp}.
\end{proof}
We conclude this section with an example of parallel composition. Suppose $G$ contains two disconnected components on nodes $S$ and $\dom \setminus S$. The set of all secret pairs correspond to pairs of values that come either from $S$ or from $\dom \setminus S$. Suppose we know two count constraints $q_S$ and $q_{\dom \setminus S}$ that count the number of tuples with values in $S$ and $\dom \setminus S$, respectively. It is easy to see that $crit(q_S) = crit(q_{\dom \setminus S}) = 0$. Therefore, running an $(\epsilon, (\dom, G, \{q_S, q_{\dom \setminus S}\}))$-Blowfish private mechanism on disjoint subsets results in no loss of  privacy.

\xxx{
\begin{proof}
Let $M_{M_1,M_2}$ denote the mechanism that outputting the results of $M_1$ and $M_2$ sequentially.
As $M_1$ satisfies $(\epsilon_1,P)$-Blowfish privacy, for every pair of neighboring databases $(D_a,D_b)\in N(P)$,
and every result $r_1\in range(M_1)$, we have  
\begin{eqnarray}
 Pr[M_1(D_a)=r_1] \leq e^{\epsilon_1}Pr[M_1(D_b)=r_1]
\end{eqnarray}
The result of $M_1$ is outputted before the result of $M_2$,
so $r_1$ will turn out to be another input of $M_2$, together with the original dataset.
As $M_1$ satisfies $(\epsilon_1,P)$-Blowfish privacy,
for every pair of neighboring databases $(D_a,D_b)\in N(P)$ coupling with the same $r_1$, 
and for every result $r_2\in range(M_2)$, we have
\begin{eqnarray}
 Pr[M_2(D_a,r_1)=r_2] \leq e^{\epsilon_1}Pr[M_2(D_b,r_1)=r_2]
\end{eqnarray}
Therefore, for every pair of neighboring databases $(D_a,D_b)\in N(P)$,
and every set of output sequence $(r1,r2)$, we have  
\begin{eqnarray}
&& Pr[M_{M_1,M_2}(D_a)=(r_1,r_2)] \nonumber \\
&=& Pr[M_1(D_a)=r_1]Pr[M_2(D_a,r_1)=r_2] \nonumber \\
&\leq& e^{\epsilon_1}Pr[M_1(D_b)=r_1]e^{\epsilon_2}Pr[M_2(D_b,r_1)=r_2] \nonumber \\
&\leq& e^{\epsilon_1+\epsilon_2}Pr[M_1(D_b)=r_1]Pr[M_2(D_b,r_1)=r_2] \nonumber \\
&=& e^{\epsilon_1+\epsilon_2} Pr[M_{M_1,M_2}(D_b)=(r_1,r_2)]
\end{eqnarray}
\end{proof}

\ \\
{\bf Parallel Composition} 
\begin{definition}[Critical Secret Pair]
Given $P=(\mathcal{T},G,\mathcal{I}_Q)$, 
a discriminative secret pair, $(s,s') \in S^G_{pairs}$ is critical to a constraint $q\in Q$
if $(s,s')\in T(D_a,D_b)$ and  $D_a \vdash q$, but $D_b \not\vdash q$. 
In addition, if $(s,s')$ is critical to $q$
and a constraint set $Q_i\subseteq Q$ contains $q$,  
then we will say $(s,s')$ is critical to the constraint set $Q_i$. 
\end{definition} 
\ \\
We denote the collection of all secret pairs critical to $Q_i$, where $Q_i\subseteq Q$,
as $criticalSP(Q_i)$, i.e.
\begin{eqnarray}
criticalSP(Q_i) = \left\{  (s,s')\in S^G_{pairs} | (s,s') \mbox{ is critical to q}, \forall q\in Q_i \right\}
\end{eqnarray}
Given a database $D$, and for any $D_i\subseteq D$, 
we denote $SP(D_i)$ as a collection of secret pairs which are related to $D_i$ (they share the same tuple id), i.e.
\begin{eqnarray}
SP(D_i) = \left\{ (s,s')\in S^G_{pairs} | s.id \in D_i \right\}
\end{eqnarray}
For each $D_i\subseteq D$, 
we denote $Q_i$ as a collection of constraints in $Q$ which the secret pairs in $SP(D_i)$ are critical to, i.e.
\begin{eqnarray}
Q_i = \left\{ q\in Q | (s,s') \mbox{ is critical to } q, \forall (s,s')\in SP(D_i) \right\}
\end{eqnarray}

\begin{definition}[Mutually Exclusive Constraint Sets]
Given $P=(\mathcal{T},G,\mathcal{I}_Q)$, 
we say constraint sets, $Q_1,Q_2,..\subseteq Q$,  are mutually exclusive, if
$criticalSP(Q_i)\cap criticalSP(Q_j)=\emptyset$ for all $i\neq j$, i.e. 
each secret pair is only critical to at most 1 constraint set $Q_i$.  
\end{definition}

\begin{theorem}[Parallel Composition] Given $P=(\mathcal{T},G,\mathcal{I}_Q)$,
consider a sequence of mechanism $M_i$ with disjoint input dataset $D_{\_id_i}$ for $i=1,2,...$,
if the constraint sets $Q_i$ corresponding to $D_{\_id_i}$ are mutually exclusive 
and $\cup_i Q_i = Q$, 
and each mechanism satisfy $(\epsilon,P_i)$-Blowfish privacy, 
where $P_i=(\mathcal{T},G,\mathcal{I}_{Q_i})$, 
then outputting the answers to all mechanisms $M_i$ for $i=1,2,...$ satisfies 
$(\epsilon,P)$-Blowfish privacy, where P= $(\mathcal{T},G,\mathcal{I}_Q)$.
\end{theorem}

\ \\
Before we show the proof of Parallel Composition, we would like to prove the lemma below.
\begin{lemma}
Given $P=(\mathcal{T},G,\mathcal{I}_Q)$,
consider a sequence of mechanism $M_i$ with disjoint input dataset $D_{\_id_i}$  for $i=1,2,...$,
each mechanism $M_i$ has its corresponding constraint set $Q_i$ mutually exclusive to each other. 
For each pair of $(D_a,D_b) \in N(P)$, 
either $T(D_a,D_b)\cap criticalSP(Q) = \emptyset$ 
or $T(D_a,D_b)\subseteq criticalSP(Q_i)$ for some $i$.
In other words, the secret pairs in $T(D_a,D_b)$ are only critical to at most one constraint set $Q_{i^*}$ and hence the corresponding $M_{i^*}$. 
\end{lemma} 
\begin{proof}
We will prove this lemma by contradiction. 
Suppose that $T(D_a,D_b)\cap criticalSP(Q) \neq \emptyset$  and there is no $Q_{i^*}$ such that $T(D_a,D_b) \subseteq criticalSP(Q_{i^*})$.
In other words,
there exists $Q_1\neq Q_2$ such that $criticalSP(Q_1)\cap T(D_a,D_b) \neq \emptyset$ and $criticalSP(Q_2)\cap T(D_a,D_b) \neq \emptyset$.

Let us create a new database $D_c$ as the same copy of $D_b$, except:
if $(s,s')\in criticalSP(Q_2)\cap T(D_a,D_b)$ and $s(D_a)=true \wedge s'(D_b)=true$, 
then $s(D_c)=true$ and hence $s'(D_c)=false$. 
Hence, $T(D_a,T_c)=T(D_a,D_b)\backslash criticalSP(Q_2)$ and then $T(D_a, D_c)\subset T(D_a,D_b)$.

In this way, we have $s(D_a)=s(D_c)$ and $s'(D_a)=s'(D_c)$ for each $(s,s')\in criticalSP(Q_2)$. 
As for any $q\in Q_2$, $D_a\vdash q$, we have $D_c\vdash q$.
Comparing $D_c$ with $D_b$, only secret pairs critical to $Q_2$ are changed, 
so for any $q\in Q\backslash Q_2$, we have $D_c \vdash q$.

Hence, $(D_a,D_b)\notin N(P)$. This leads to contradiction.
Therefore, the secret pairs in $T(D_a,D_b)$ are only critical to at most one constraint set, $Q_{i^*}$, and its corresponding $M_{i^*}$.
\end{proof}

\begin{lemma}
If $Q_{i^*}$ in the previous lemma exists, then 
\begin{enumerate}
 \item for $i\neq i^*$, $D_{a\_id_i} = D_{b\_id_i}$; 
 \item $(D_{a\_id_{i^*}}, D_{a\_id_{i^*}})\in N(P_{i^*})$.
\end{enumerate}
\end{lemma}
\begin{proof}
The first point follows the proof of the previous lemma. We would like to show the second point.
It is true that for  $D_{a\_id_{i^*}}$, $D_{a\_id_{i^*}}$ :
\begin{enumerate}
 \item $D_{1\_id_{i^*}}, D_{2\_id_{i^*}}\in I_{Q_1}$;
 \item $T(D_{1\_id_{i^*}}, D_{2\_id_{i^*}}) = T(D_1,D_2)\neq \emptyset$
\end{enumerate}
In order to prove $(D_{a\_id_{i^*}}, D_{a\_id_{i^*}})\in N(P_{i^*})$, we need to show the last condition is true. 
Suppose the last condition is not true, i.e. there exist $D_{c\_id_{i^*}} \vdash Q_{i^*}$, such that
\begin{enumerate}
 \item[(a)] $T(D_{a\_id_{i^*}},D_{c\_id_{i^*}}) \subset T(D_{a\_id_{i^*}},D_{b\_id_{i^*}})$, or
 \item[(b)] $T(D_{a\_id_{i^*}},D_{c\_id_{i^*}}) =T(D_{a\_id_{i^*}},D_{b\_id_{i^*}})$ and \\
$\triangle (D_{a\_id_{i^*}},D_{c\_id_{i^*}}) \subset \triangle(D_{a\_id_{i^*}},D_{b\_id_{i^*}})$
\end{enumerate}
Then this condition is not true for $(D_a,D_b)$, as we could find a $D_c$ from $D_{c\_id_{i^*}}$ 
which satisfies (a) and (b) by adding tuples $t$ in $D_a$ with $t.id \neq \_id_{i^*}$.
Then we get $(D_a,D_b)\notin N(P)$. This leads to contradiction.
Therefore,  $(D_{a\_id_{i^*}}, D_{b\_id_{i^*}})\in N(P_{i^*})$.
\end{proof}

\ \\
\begin{proof}[Proof for Parallel Composition] 
As each $M_i$ satisfies $(\epsilon,P_i)$-Blowfish privacy, 
for every pair of neighboring databases $(D_{a\_id_i},D_{b\_id_i})\in N(P_i)$,
and every output $r_i\subseteq range(M_i)$, we have  
\begin{eqnarray}
 Pr[M_i(D_{a\_id_i})=r_i] \leq e^{\epsilon}Pr[M_i^r(D_{b\_id_i})=r_i]
\end{eqnarray}
If $D_{a\_id_i}=D_{b\_id_i}$, then 
\begin{eqnarray}
 Pr[M_i(D_{a\_id_i})=r_i] = Pr[M_i^r(D_{b\_id_i})=r_i]
\end{eqnarray}
Therefore, for every pair of neighboring databases $(D_a,D_b)\in N(P)$,
and every set of output sequence $r$, by the previous lemma, we have 
\begin{eqnarray}
&& Pr[M(D_a)=r] \nonumber \\
&=& \prod_i Pr[M_i(D_{a\_id_i})=r_i] \nonumber \\
&\leq&  e^{\epsilon} Pr[M_{i^*}(D_{b\_id_{i^*}})=r_{i^*}] \prod_{i,i\neq i^*}Pr[M_i^r(D_{b\_id_i})=r_i] \nonumber \\
&=& e^{\epsilon}\prod_iPr[M_i(D_{b\_id_i})=r_i] \nonumber \\
&=& e^{\epsilon}Pr[M(D_b)=r]
\end{eqnarray}
\end{proof}
}

\subsection{Relation to other definitions}
\label{sec:related}
In this section, we relate Blowfish privacy to existing notions of privacy. 
We discuss variants of differential privacy \cite{icalp:Dwork06} (including restricted sensitivity \cite{Blocki13:restricted}), the Pufferfish framework \cite{pods:KiferM12}, privacy axioms \cite{Kifer10axioms}, and a recent independent work on extending differential privacy with metrics \cite{pets13:metric}. 

\newstuff{\stitle{Differential Privacy \cite{icalp:Dwork06}:}}
 One can easily verify that a mechanism satisfies $\epsilon$-differential privacy (Definition~\ref{def:dp}) if and only if it satisfies $(\epsilon, P)$-Blowfish privacy, where $P = (\dom, K, \inp_n)$, and $K$ is the complete graph on the domain. \newstuff{Thus, Blowfish privacy is a generalization of differential privacy that allows a data curator to trade-off privacy vs utility by controlling sensitive information $G$ (instead of $K$) and auxiliary knowledge $\inp_Q$ (instead of $\inp_n$) in the policy. }

\newstuff{\stitle{Pufferfish Framework \cite{pods:KiferM12,tods:KiferM13}:}}
 Blowfish borrows the sensitive information specification from Pufferfish. Pufferfish defines adversarial knowledge using a set of data generating distributions, while Blowfish instantiates the same using publicly known constraints. We can show formal relationships between Blowfish and Pufferfish instantiations. 
\begin{theorem}\label{thm:puff1}
Let $\Spairs$ be the set of discriminative pairs corresponding to policy $P = (\dom, G, \inp_n)$. Let ${\cal D}$ denote the set of all product distributions $\{p_i(\cdot)\}_i$ over $n$ tuples.  $p_i(\cdot)$ denotes a probability distribution for tuple $i$ over  $\dom$. Then a mechanism satisfies $(\epsilon, \Spairs, {\cal D})$-Pufferfish privacy if and only if it satisfies $(\epsilon, P)$-Blowfish privacy.
\end{theorem}
\begin{theorem}\label{thm:puff2}
Consider a policy $P = (\dom, G, \inp_Q)$ corresponding to a set of constraints $Q$. Let $\Spairs$ be defined as in Theorem~\ref{thm:puff1}. Let ${\cal D}_Q$ be the set of product  distributions conditioned on the constraints in $Q$; i.e., 
\[ P[D = {x_1, \ldots, x_k}] \propto \left\{\begin{array}{ll}\prod_{i} p_i(x_i) & \mbox{ if $D \in \inp_Q$} \\ 0 & \mbox { otherwise}\end{array}\right.\]
A mechanism $M$ that satisfies $(\epsilon, \Spairs, {\cal D}_Q)$-Pufferfish privacy also satisfies $(\epsilon, P)$-Blowfish privacy.
\end{theorem}
\newstuff{Theorem~\ref{thm:puff1} states that Blowfish policies without constraints are equivalent to  Pufferfish instantiated using adversaries who believe tuples in $D$ are independent (proof follows from Theorem 6.1 \cite{tods:KiferM13}). Theorem~\ref{thm:puff2} states that when constraints are known, Blowfish is a necessary condition for any mechanism that satisfies a similar Pufferfish instantiation with constraints (we conjecture the sufficiency of Blowfish as well). Thus Blowfish privacy policies correspond to a subclass of privacy definitions  that can be instantiated using  Pufferfish. 

Both Pufferfish and Blowfish aid the data publisher to customize privacy definitions by carefully defining sensitive information and adversarial knowledge. However, Blowfish improves over Pufferfish in three key aspects. First, there are no general algorithms known for Pufferfish instantiations. In this paper, we present of algorithms for various Blowfish policies. Thus, we can't compare Blowfish and Pufferfish experimentally. Second, all Blowfish privacy policies result in composable privacy definitions. This is not true for the Pufferfish framework. Finally, we believe Blowfish privacy is easier to understand and use than the Pufferfish framework for data publishers who are not privacy experts.\footnote{\small We have some initial anecdotal evidence of this fact working with statisticians from the US Census.} For instance, one needs to specify adversarial knowledge as sets of complex probability distributions in Pufferfish, while in Blowfish policies one only needs to specify conceptually simpler publicly known constraints.} 

\newstuff{\stitle{Other Privacy Definitions:}}
Kifer and Lin \cite{Kifer10axioms} stipulate that every ``good'' privacy definition should satisfy two axioms -- transformation invariance, and convexity. We can show that  Blowfish privacy satisfy both these axioms.

Recent papers have extended differential privacy to handle constraints. Induced neighbor privacy \cite{sigmod:KiferM11, pods:KiferM12} extends the notion of neighbors such that neighboring databases satisfy the constraints and are minimally far apart (in terms of tuple changes). Blowfish extends this notion of induced neighbors to take into account discriminative pairs of secrets and measures distance in terms of the set of different discriminative pairs. Restricted sensitivity \cite{Blocki13:restricted} extends the notion of sensitivity to account for constraints. In particular, the restricted sensitivity of a function $f$ given a set of constraints $Q$, or $RS_f(Q)$, is the maximum $|f(D_1) - f(D_2)| / d(D_1, D_2)$, over all $D_1, D_2 \in \inp_Q$. However, tuning noise to $RS_f(Q)$ may not limit the ability of an attacker to learn sensitive information. For instance, if $\inp_Q = \{0^n, 1^n\}$, then the restricted sensitivity of releasing the number of 1s is $1$. Adding constant noise does not disallow the adversary from knowing whether the database was $0^n$ or $1^n$.

A very recent independent work suggests extending differential privacy using a metric over all possible databases \cite{pets13:metric}. In particular, given a distance metric $d$ over instances, they require an algorithm to ensure that $P[M(X) \subseteq S] \leq e^{\epsilon\cdot d(X,Y)}P[M(Y) \subseteq S]$, for all sets of outputs $S$ and all instances $X$ and $Y$. Thus differential privacy corresponds to a specific distance measure -- Hamming distance. The sensitive information specification in Blowfish can also be thought of in terms of a distance metric over tuples. In addition we present novel algorithms (ordered mechanism) and allow incorporating knowledge of constraints. We defer a more detailed comparison to future work.



\section{Blowfish without Constraints}
\label{sec:noconstraints}
Given any query $f$ that outputs a vector of reals, we can define a {\em policy specific sensitivity} of $f$. Thus, the Laplace mechanism with noise calibrated to the policy specific sensitivity ensures Blowfish privacy.
\begin{definition}\label{def:policyGS} 
{\sc (Policy Specific Global Sensitivity).} Given a policy $(\dom, G, \inp_Q)$, $S(f,P)$ denotes the policy specific global sensitivity of a function $f$ and is defined as $\max_{(D_1, D_2) \in N(P)} ||f(D_1) - f(D_2)||_1$.
\end{definition}
\begin{theorem}\label{thm:lappolicy}
Let $P = (\dom, G, \inp_Q)$ be a policy. Given a function $f:\inp_Q \rightarrow \mathbb{R}^d$, outputting $f(D) + \mathbf{\eta}$ ensures $(\epsilon, P)$-Blowfish privacy if $\mathbf{\eta} \in \mathbb{R}^d$ is a vector of independent random numbers drawn from $Lap(S(f,p)/\epsilon)$.
\end{theorem}

\noindent{\bf 
When policies do not have constraints} ($P = (\dom, G, \inp_n)$), $(\epsilon, P)$-Blowfish differs from $\epsilon$-differential privacy only in the specification of sensitive information. Note that every pair $(D_1, D_2) \in N(P)$ differ in only one tuple when $P$ has no constraints. Therefore, the following result trivially holds.  
\begin{lemma}
Any mechanism $M$ that satisfies\\ $\epsilon$-differential privacy also satisfies $(\epsilon, (\dom, G, \inp_n))$-Blowfish privacy for all discriminative secret graphs $G$. 
\end{lemma}
\noindent The proof follows from the fact that $\epsilon$-differential privacy is equivalent to $(\epsilon, (\dom, K, \inp_n))$-Blowfish privacy, where $K$ is the complete graph.

In many cases, we can do better in terms of utility than differentially privacy mechanisms. It is easy to  see that $S(f, P)$ is never larger than the global sensitivity $S(f)$. Therefore, just using the Laplace mechanism with $S(f,P)$ can provide better utility.

For instance, consider a linear sum query $f_{\mathbf{w}} = \sum_{i=1}^n w_i x_i$, where $\mathbf{w} \in \mathbb{R}^n$ is a weight vector, and each value $x_i \in \dom = [a,b]$. For $G^{\rm full}$, the policy specific sensitivity is $(b-a) \cdot (\max_i w_i)$ the same as the global sensitivity. For $G^{d, \theta}$, where $d(x,y) = |x-y|$, the policy specific sensitivity is $\theta\cdot (\max_iw_i)$, which can be much smaller than the global sensitivity when $\theta \ll (b-a)$. 

As a second example, suppose ${\cal P}$ is a partitioning of the domain. If the policy specifies sensitive information partitioned by ${\cal P}$ ($G^{\cal P}$), then the policy specific sensitivity of $h_{\cal P}$ is 0. That is, the histogram of ${\cal P}$ or any coarser partitioning can be released without any noise. We will show more examples of improved utility under Blowfish policies in Sec~\ref{sec:noconstraints}.  

However, for histogram queries, the policy specific sensitivity for most reasonable policies (with no constraints) is $2$, the same as global sensitivity.\footnote{\small The one exception is partitioned sensitive information.} Thus, it cannot significantly improve the accuracy for histogram queries.

\newstuff{
Next, we present examples of two analysis tasks -- $k$-means clustering (Section~\ref{sec:kmeans}), and releasing cumulative histograms (Section~\ref{sec:cdf})-- for which we can design mechanisms for Blowfish policies without constraints with more utility (lesser error) than mechanisms that satisfy differential privacy. In $k$-means clustering we will see that using Blowfish policies helps reduce the sensitivity of intermediate queries on the data. In the case of the cumulative histogram workload, we can identify novel query answering strategies given a Blowfish policy that helps reduce the error. 
}


\section{K-means Clustering}
\label{sec:kmeans}
\newcommand{\skin}{{\sc skin segmentation}\xspace}
\newcommand{\BGR}{{\sf B, G, R}\xspace}
\newcommand{\adult}{{\sc adult}\xspace}
\newcommand{\tweet}{{\sc twitter}\xspace}
\newcommand{\caploss}{{\sf capital loss}\xspace}
\newcommand{\work}{{\sf work class}\xspace}
\newcommand{\sex}{{\sf sex}\xspace}
\renewcommand{\lat}{{\sf latitude}\xspace}
\renewcommand{\long}{{\sf longitude}\xspace}

\begin{figure*}[t]
\centering
\subfigure[\tweet, $G^{L_1, \theta}$]{
\includegraphics[scale=0.6]{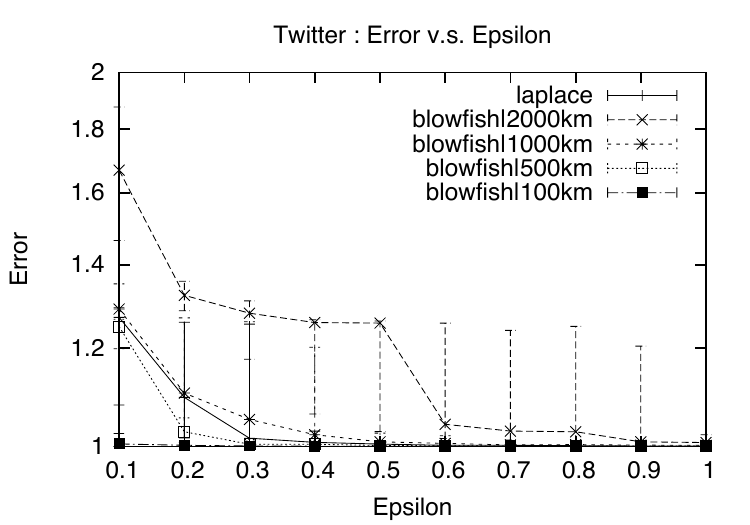}
\label{fig:kmeans_twitter}
}
\subfigure[\skin, $G^{L_1, \theta}$]{
\includegraphics[scale=0.6]{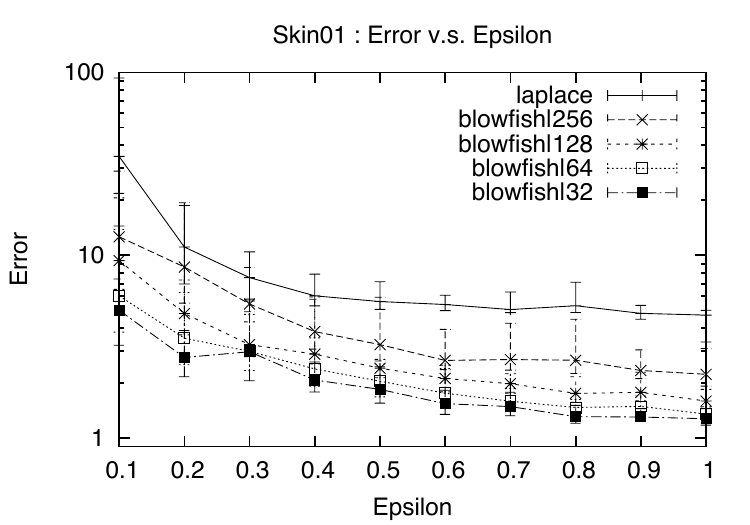}
\label{fig:kmeans_skin01}
}
\subfigure[Synthetic data set, $G^{L_1, \theta}$]{
\includegraphics[scale=0.6]{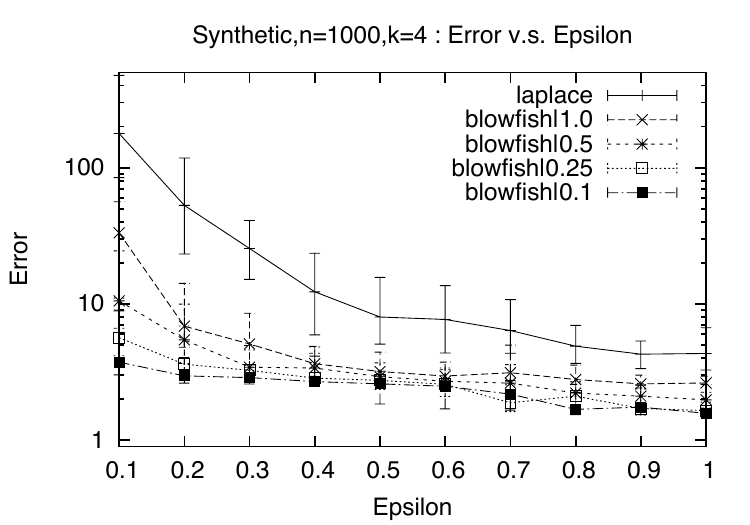}
\label{fig:kmeans_synth_1000_4}
}
\newstuff{
\subfigure[\skin, $G^{L_1, \theta}$]{
\includegraphics[scale=0.6]{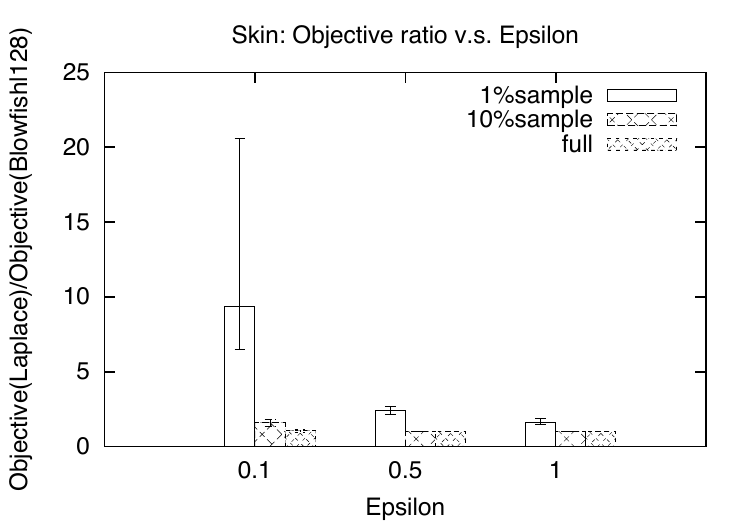}
\label{fig:kmeans_skinall}
}
\subfigure[All datasets, $G^{\mathrm{attr}}$]{
\includegraphics[scale=0.6]{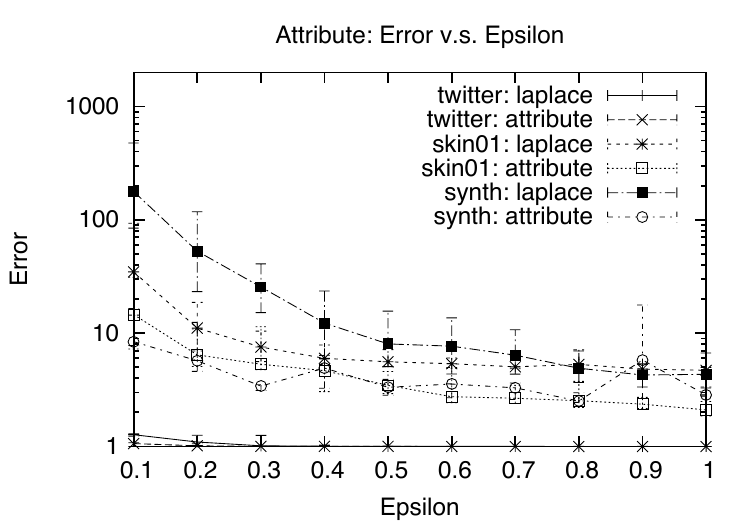}
\label{fig:kmeans_attr}
}
\subfigure[\tweet, $G^{\cal P}$]{
\includegraphics[scale=0.6]{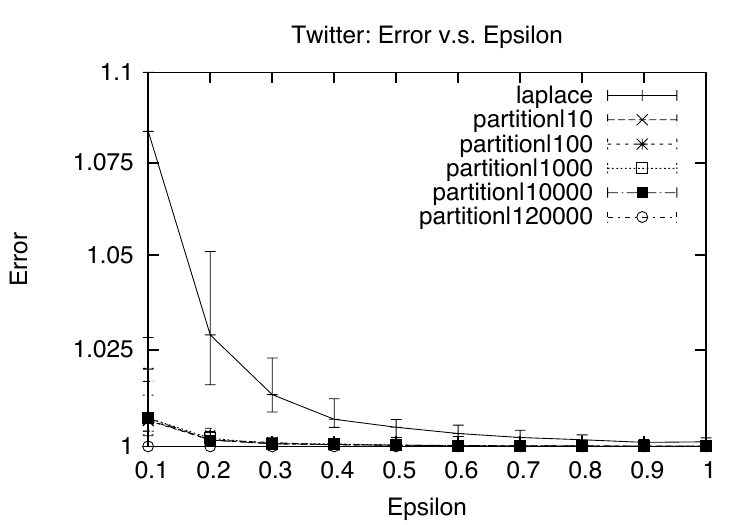}
\label{fig:kmeans_partition}
}
}
\caption{\label{fig:kmeans}K-Means: Error under Laplace mechanism v.s Blowfish privacy for different discriminative graphs.}
\end{figure*}

$K$-means clustering is widely used in many applications such as classification and feature learning. It aims to cluster proximate data together and is formally defined below.
\begin{definition}[$K$-means clustering]
Given a data set of $n$ points $(t_1,...,t_n) \in \dom^n$, $k$-means clustering aims to partition the points into $k \leq n$ clusters $S=\left\{ S_1,...,S_k\right\}$ in order to minimize  
\begin{equation} \label{eq:kmeans}
\sum_{i=1}^{k}\sum_{t_j\in S_i} ||t_j-\mu _i ||^2,
\end{equation}
where $\mu_i = \frac{1}{|S_i|}\sum_{t_j\in S_i} t_j$, and $||x-y||$ denotes $L_2$ distance.
\end{definition}

The non-private version of $k$-means clustering initializes the means/centroids $(\mu_1,...,\mu_k)$ (e.g. randomly) and updates them iteratively as follows:
1) assign each point to the nearest centroid; 
2) recompute the centroid of each cluster,
until reaching some convergence criterion or a fixed number of iterations.

The first differentially private $k$-means clustering algorithm was proposed by Blum et al.~\cite{pods:blum05} as SuLQ $k$-means. Observe that only two queries are required explicitly: 1) the number of points in each new cluster, $q_{size} = (|S_1|,...,|S_k|)$  and 2) the sum of the data points for each cluster, $q_{sum} = (\sum_{t_j\in S_1}t_j,...,\sum_{t_j\in S_k}t_j)$, to compute the centroid. The sensitivity of $q_{size}$ is $2$ (same as a histogram query). \newstuff{Let $d(\dom)$ denote the diameter of the domain, or the largest $L_1$ distance ($||x-y||_1$) between any two points $x, y \in \dom$. The sensitivity of $q_{sum}$ could be as large as the diameter $2\cdot d(\dom)$ since a tuple from $x$ to $y$ can only change the sums for two clusters by at most $d(\dom)$}. 

Under Blowfish privacy policies, the policy specific sensitivity of $q_{sum}$ can be much smaller than $|\dom|$ under Differential privacy (i.e. complete graph $G^{full}$ for Blowfish policies). Since $q_{size}$ is the histogram query, the sensitivity of $q_{size}$ under Blowfish is also $2$. 
\newstuff{
\begin{lemma}\label{lemma:sensitivity}
Policy specific global sensitivities of $q_{sum}$ under the attribute $G^{\mathrm{attr}}$,  $L_1$-distance $G^{(L_1, \theta)}$, and partition $G^{\cal P}$ discriminative graphs (from Section~\ref{sec:policy}) are smaller than the global sensitivity of $q_{sum}$ under differential privacy.
\end{lemma}
\begin{proof}
First, in the attribute discriminative graph $G^{\mathrm{attr}}$, edges correspond to  $(x, y) \in \dom \times \dom$ that differ only in any one attribute. Thus, if $|A|$ denotes maximum distance between two elements in $A$, then the policy specific sensitivity of $q_{sum}$ under $G^{\mathrm{attr}}$ is $\max_A (2\cdot |A|) < 2 \cdot d(\dom)$. 
Next, suppose we use $G^{L_1, \theta}$, where  $x, y \in \dom$ are connected by an edge if $||x-y||_1 \leq \theta$. Thus, policy specific sensitivity of $q_{sum}$ is $2\theta$.

Finally, consider the policy specified using the partitioned sensitive graph $G^{\cal P}$, where ${\cal P} = \{P_1, P_2, \ldots, P_k\}$ is some data independent partitioning of the domain $\dom$. Here, an adversary should not distinguish between an individual's tuple taking a pair of values $x, y \in \dom$  only if $x$ and $y$ appear in the same partition $P_i$ for some $i$. Under this policy the sensitivity of $q_{sum}$ is at most $\max_{P \in {\cal P}} 2 \cdot d(P) < 2 \cdot d(\dom)$.
\end{proof}

Thus, by Theorem~\ref{thm:lappolicy}, we can use the SULQ $k$-means mechanism with the appropriate policy specific sensitivity for $q_{sum}$ (from Lemma~\ref{lemma:sensitivity}) and thus satisfy privacy under the Blowfish policy while ensuring better accuracy. }

\newstuff{
\subsection{Empirical Evaluation}\label{sec:kmeanseval}
}
We empirically evaluate the accuracy of $k$-means clustering for $(\epsilon, (\dom, G, \inp_n))$-Blowfish privacy on three data sets. The first two datasets are real-world datasets -- \tweet and \skin \footnote{{\small\url{http://archive.ics.uci.edu/ml/datasets/Skin+Segmentation}}}. The \tweet data set consists of a total of $193563$ tweets collected using Twitter API that all contained a \lat/\long within a bounding box of $50N,125W$ and $30N,110W$ (western USA) -- about $2222\times 1442$ square km. By setting the precision of \lat/\long coordinates to be 0.05, we obtain a 2D domain of size $400 \times 300$. 
The \skin data set consists of $245057$ instances. 
Three ordinal attributes are considered and they are \BGR values from face images of different classes. Each of them has a range from $0$ to $255$. To understand the effect of Blowfish policies on datasets of different sizes, we consider the full dataset {\sc  skin}, as well as a 10\% and 1\% sub-sample ({\sc skin10}, {\sc skin01}) of the data. 

The third dataset is a synthetic dataset where we generate 1000 points from $(0,1)^4$ with $k$ randomly chosen centers and a Gaussian noise with $\sigma(0,0.2)$ in each direction. 

In Figures \ref{fig:kmeans_twitter}-\ref{fig:kmeans_skinall}, we report the ratio of the mean of the objective value in Eqn. (\ref{eq:kmeans}) between private clustering methods including Laplace mechanism and Blowfish privacy with $\Spairs^{\dist, \theta}$, and the non-private k-means algorithm, for various values of $\epsilon=\left\{0.1,0.2,...,0.9,1.0\right\}$. For all datasets, $d(\cdot)$ is $L_1$ (or Manhattan) distance. The number of iterations is fixed to be 10 and the number of clusters is $k=4$. Each experiment is repeated 50 times to find mean, lower and upper quartile. 
Figure \ref{fig:kmeans_twitter} clusters according to \lat/\long of each tweet. We consider five different policies: $G^{full}$ (Laplace mechanism), $G^{d,2000km}$, $G^{d,1000km}$, $G^{d,500km}$, $G^{d,100km}$. Here, $\theta=100$km means that the adversary cannot distinguish locations within a 20000 square km region. Figure \ref{fig:kmeans_skin01} clusters the 1\% subsample {\sc skin01} based on three attributes: \BGR values and considers 5 policies as well: $G^{full}$, $G^{d,256}$, $G^{d,128}$, $G^{d,64}$ and $G^{d,32}$. Lastly, we also consider five policies for the synthetic dataset in Figure \ref{fig:kmeans_synth_1000_4}: $G^{full}$, $G^{d,1.0}$, $G^{d,0.5}$, $G^{d,0.25}$, $G^{d,0.1}$. 

From Figures \ref{fig:kmeans_twitter}-\ref{fig:kmeans_synth_1000_4}, we observe that the objective value of Laplace mechanism could deviate up to 100 times away from non-private method, but under Blowfish policies objective values could be less than 5 times that for non-private k-means. Moreover, the error introduced by Laplace mechanism becomes larger with higher dimensionality -- the ratio for Laplace mechanism in Figure Figure \ref{fig:kmeans_synth_1000_4} and \ref{fig:kmeans_skin01} (4 and 3 dimensional resp.) is much higher than that in the 2D \tweet dataset in Figure \ref{fig:kmeans_twitter}. From Figure \ref{fig:kmeans_skin01}, we observe that the error introduced by private mechanisms do not necessarily reduce monotonically as we reduce Blowfish privacy protection (i.e. reduce $\theta$). The same pattern is observed in Figure \ref{fig:kmeans_twitter} and Figure \ref{fig:kmeans_synth_1000_4}. One possible explanation is that adding a sufficient amount of noise could be helpful to get out of local minima for clustering, but adding too much noise could lead to less accurate results. 

\newstuff{
To study the interplay between dataset size and Blowfish, we plot (Figure~\ref{fig:kmeans_skinall}) for {\sc skin}, {\sc skin10} and {\sc skin01} the ratio of the objective value attained by the Laplace method to the objective value attained by one of the Blowfish policies: $G^{d, 128}$. In all cases, we see an improvement in the objective under Blowfish. The improvement in the objective is smaller for larger $\epsilon$ and larger datasets (since the Laplace mechanism solution is close to the non-private solutions on {\sc skin}).

Finally, Figures~\ref{fig:kmeans_attr} and \ref{fig:kmeans_partition} summarize our results on the $G^{\mathrm{attr}}$ and $G^{\cal P}$ discriminative graphs. Figure~\ref{fig:kmeans_attr} shows that under the $G^{\mathrm{attr}}$ Blowfish policy, the error decreases by an order of magnitude compared to the Laplace mechanism for {\sc skin01} and the synthetic dataset due to higher dimensionality and small dataset size. On the other hand, there is little gain by using $G^{\mathrm{attr}}$ for the larger 2D \tweet dataset.

Figure~\ref{fig:kmeans_partition} shows ratio of the objective attained by the private methods to that of the non-private k-means under $G^{\cal P}$, for partitions ${\cal P}$ of different sizes. In each case, the 300x400 grid is uniformly divided; e.g., in $\tt{partition|100}$, we consider a uniform partitioning in 100 coarse cells, where each new cell contains 30x40 cells from the original grid. Thus an adversary will not be able to tell whether an individual's location was within an area spanned by the 30x40 cells (about 36,300 sq km). {\tt partition|120000} corresponds to the original grid; thus we only protects pairs of locations within each cell in the original grid (about 30 sq km). We see that the objective value for Blowfish policies are smaller than the objective values under Laplace mechanisms, suggesting more accurate clustering. We also note that under {\tt partition|120000}, we can do the clustering exactly, since the sensitivity of both $q_{size}$ and $q_{sum}$ are 0. 

To summarize, Blowfish policies allow us to effectively improve utility by trading off privacy. In certain cases, we observe that Blowfish policies attain an objective value that is close to 10 times smaller than that for the Laplace mechanism. The gap between Laplace and Blowfish policies increases with dimensionality, and reduces with data size.  
}

\section{Cumulative Histograms}
\label{sec:cdf}
In this section, we develop novel query answering strategies for two workloads -- cumulative histograms and range queries.
Throughout this section, we will use Mean Squared Error as a measure of accuracy/error defined in Def~\ref{def:mse}. 
\begin{definition}[Cumulative Histogram]
Consider a domain $\dom = \left\{x_1,...,x_{|\dom |} \right\}$ that has a total ordering $x_1\leq...\leq x_{|\dom |}$. Let $c(x_i)$ denote the number of times $x_i$ appears in the database $D$. Then, the cumulative histogram of $\dom $, denoted by $S_{\dom}(\cdot)$ is a sequence of cumulative counts
\begin{equation}
\left\{ s_i \mid s_i=\sum_{j=1}^i c(x_j), \forall i=1,...,|\dom |  \right\}
\end{equation}
\end{definition}

Since we know the total size of the dataset $|D|=n$, dividing each cumulative count in $S_{\dom}(\cdot)$ by $n$ gives us the cumulative distribution function (CDF) over $\dom$. Releasing the CDF has many applications including computing quantiles and histograms, answering range queries and constructing indexes (e.g. $k$-d tree). This motivates us to design a mechanism for releasing cumulative histograms.

The cumulative histogram has a global sensitivity of $|\dom |-1$ because all the counts in cumulative histogram except $s_{|\dom |}$ will be reduced by 1 when a record in $D$ changes from $x_1$ to $x_{|\dom |}$. Similar to $k$-means clustering, we could reduce the sensitivity of cumulative histogram $S_{\dom}(\cdot)$ by specifying the sensitive information, such as $\Spairs^{\cal P}$ and $\Spairs^{\dist, \theta}$. For this section, we focus on $\Spairs^{\dist, \theta}$, where $d(\cdot)$ is the $L1$ distance on the domain and we assume that all the domains discussed here have a total ordering. 

\subsection{Ordered Mechanism}
Let us first consider a policy $P_\theta=(\dom,G^{d,\theta},\inp_n)$ with $\theta=1$. The discriminative secret graph is a line graph, $G^{d,1}=(V,E)$, where $V=\dom $ and $E=\left\{(x_i,x_{i+1})|\forall i=1,...,|\dom|-1 \right\}$. This means that only adjacent domain values $(x_i,x_{i+1}) \in \dom \times \dom $ can form a secret pair. Therefore, the policy specific sensitivity of $S_{\dom}(\cdot)$ for a line graph is 1. Based on this small sensitivity, we propose a mechanism, named Ordered Mechanism $M_{G^{d,1}}^O$ to perturb cumulative histogram $S_{\dom}(\cdot)$ over line graph $G^{d,1}$ in the following way. For each $s_i$, we add $\eta_i \sim Laplace(\frac{1}{\epsilon})$ to get $\tilde{s}_i$ to ensure $(\epsilon, P)$-Blowfish privacy. Each $\tilde{s}_i$ has an error with an expectation equals to $\frac{2}{\epsilon^2}$. Note that Theorem~\ref{thm:lappolicy} already ensures that releasing $\tilde{s}_i$'s satisfies $(\epsilon, P_1)$-Blowfish privacy. Furthermore, observe that the counts in $S_{\dom}(\cdot)$ are in ascending order. Hence, we can boost the accuracy of $\tilde{S}_{\dom}(\cdot)$ using constrained inference proposed by Hay et al. in~\cite{vldb:HayRMS10}. 
In this way, the new cumulative histogram, denoted by $\hat{S}_{\dom}(\cdot)$, satisfies the ordering constraint and has an error $\error_{\hat{S}}=O(\frac{p\log^3|\dom |}{\epsilon^2})$, where $p$ represents the number of distinct values in $S_{\dom}(\cdot)$~\cite{vldb:HayRMS10}. Note that, if we additionally enforce the constraint that $s_1>0$, then all the counts are also positive. In particular, when $p=1$, $\error_{\hat{S}}=O(\frac{\log^3 |\dom |}{\epsilon^2})$ and when $p=|\dom |$, $\error_{\hat{S}} = O(\frac{|\dom |}{\epsilon^2})$. Many real datasets are sparse, i.e. the majority of the domain values have zero counts, and hence have fewer distinct cumulative counts, i.e. $p \ll |\dom |$. This leads to much smaller $\error_{\hat{S}}$ compared to $\error_{\tilde{S}}$. The best known strategy for releasing the cumulative histogram is using the hierarchical mechanism~\cite{vldb:HayRMS10}, which results in a total error of $O(\frac{|\dom |\log^3|\dom |}{\epsilon^2})$. Moreover, the SVD bound~\cite{edbt:chaoli13} suggests that no strategy can release the cumulative histogram with $O(\frac{|\dom|}{\epsilon^2})$ error. Thus under the line graph policy, the Ordered Mechanism is a much better strategy for cumulative histogram.

One important application of cumulative histogram is answering range query, defined as follows.   
\begin{definition}[Range Query]
Let $D$ has domain $\dom = \left\{x_1,...,x_{|\dom |} \right\}$, where $\dom$ has a total ordering. A range query, denoted by $q[x_i,x_j]$  counts the number of tuples falling within the range $[x_i,x_j]$ where $x_i,x_j\in\dom$ and $x_i\leq x_j$.
\end{definition}

Range queries can be directly answered using cumulative histogram $\hat{S}_{\dom}(\cdot)$, $q[x_i,x_j]=\hat{s}_j-\hat{s}_{i-1}$. As each range query requires at most two noisy cumulative counts, it has an error smaller than $2\cdot \frac{2}{\epsilon^2}$ (even without constrained inference). Hence, we have the following theorem.
\begin{theorem} Consider a policy ($\dom, G^{d,1}, \inp_n$), where $G^{d,1}$ is a line graph. Then the expected error of a range query $q[x_i,x_j]$ for Ordered Mechanism is given by:
\begin{equation}
 \error_{q[x_i,x_j],M_{G^{d,1}}^O} \leq 4/\epsilon^2
\end{equation}
\end{theorem}
This error bound is independent of $|\dom |$, much lower than the expected error using hierarchical structure with Laplace mechanism to answer range queries, $\error_{q[x_i,x_j],lap} = \frac{\log^3|\dom |}{\epsilon^2}$. Again, the SVD bound~\cite{edbt:chaoli13} suggests that no differentially private strategy can answer each range query with $O(\frac{1}{\epsilon^2})$ error. Other applications of $S_{\dom}(\cdot)$ including computing quantiles and histograms and constructing indexes (e.g. $k$-d tree) could also use cumulative histogram in a similar manner as range query to obtain a much smaller error by trading utility with privacy under $G^{d,1}$. Next, we describe the {\em ordered hierarchical mechanism} that works for general graphs, $G^{d,\theta}$.  


\subsection{Ordered Hierarchical Mechanism}
\label{sec:rangequery}
\begin{figure*}
\centering
\subfigure[Ordered Hierarchical Tree  $\theta = 4$]{
\includegraphics[scale=0.2]{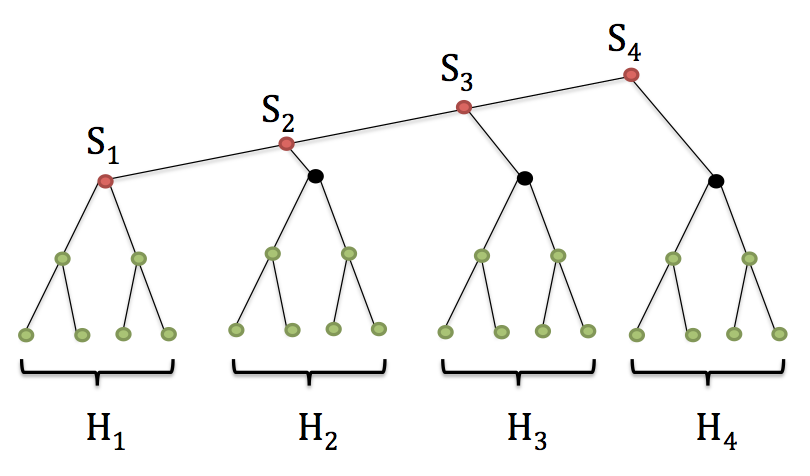}
\label{fig:oh}
}
\subfigure[Adult - capital loss]{
\includegraphics[scale=0.7]{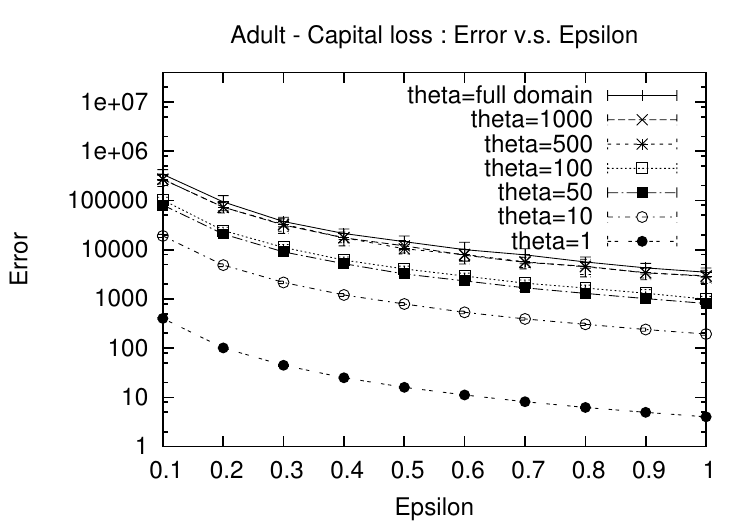}
\label{fig:range_mean_adult}
}
\subfigure[Twitter - latitude]{
\includegraphics[scale=0.7]{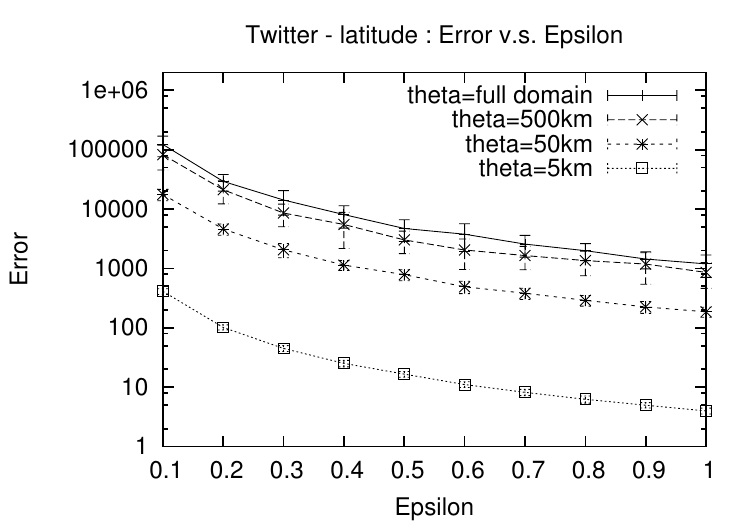}
\label{fig:range_mean_tweetlat}
}
\caption{\label{fig:range}Ordered Hierarchical Mechanism. \ref{fig:oh} gives an example of $OH$, where $\theta=4$. \ref{fig:range_mean_adult} and \ref{fig:range_mean_tweetlat} shows privacy-utility trade-offs for range query, using $G^{d,\theta}$ for sensitive information.}
\end{figure*}

For a more general graph $G^{d,\theta}=(V,E)$, where $V=\dom$ and $E=\left\{(x_i,x_{i\pm1}),..,(x_i,x_{i\pm\theta})|\forall i=1,...,|\dom | \right\}$, the sensitivity of releasing cumulative histogram $S_{\dom}(\cdot)$ becomes $\theta$. The Ordered Mechanism would add noise from $Lap(\frac{\theta}{\epsilon})$ to each cumulative counts $s_i$. The total error in the released cumulative histogram and range queries would still be asymptotically smaller than the error achieved by any differentially private mechanism for small $\theta$. However, the errors become comparable as the $\theta$ reaches $\log|\dom|$, and the Ordered Mechanism's error exceeds the error from the hierarchical mechanism when $\theta = O (\log^{3/2} |\dom |)$. In this section, we present a hybrid strategy for releasing cumulative histograms (and hence range queries), called Ordered Hierarchical Mechanism, that always has an error less than or equal to the hierarchical mechanism for all $\theta$.



Various hierarchical methods have been proposed in the literature~\cite{vldb:HayRMS10, icde:XiaoWG10, pods:LiHRMM10, icde:XuZXYY12, pvldb:qardaji13}. A basic hierarchical structure is usually described as a tree with a regular fan-out $f$. The root records the total size of the dataset $D$, i.e. the answer to the range query $q[x_1, x_{|\dom|}]$. This range is then partitioned into $f$ intervals. If $\delta = \lceil \frac{|\dom|}{f} \rceil$, the intervals are $[x_1,x_{\delta}]$, $[x_{\delta+1},x_{2\delta}]$,..., $[x_{|\dom|-\delta+1},x_{|\dom|}]$ and answers to the range queries over those intervals are recorded by the children of the root node. Recursively, the interval represented by the current node will be further divided into $f$ subintervals. Leaf nodes correspond to unit length interval $q[x_i,x_i]$. The height of the tree is $h=\lceil \log_f |\dom| \rceil$. In the above construction, the counts at level $i$  are released using the Laplace mechanism with parameter $\frac{2}{\epsilon}$, and $\sum_i \epsilon_i = \epsilon$. Prior work has considered distributing the $\epsilon$ uniformly or geometrically~\cite{icde:CormodePSSY12}. We use uniform budgeting in our experiments.

Inspired by ordered mechanism for line graph, we propose a hybrid structure, called Ordered Hierarchical Structure $OH$ for $(\epsilon, (\dom, G^{d,\theta}, \inp_n))$-Blowfish privacy. As shown in Figure~\ref{fig:oh}, $OH$ has two types of nodes, $S$ nodes and $H$ nodes. The number of $S$ nodes is $k=\ceil{\frac{n}{\theta}}$, which is dependent on the threshold $\theta$. 
In this way, we could guarantee a sensitivity of 1 among the $S$ nodes. 
Let us represent $S$ nodes as $s_1,...,s_k$, where $s_1=q[x_1,x_{\theta}]$,..., $s_{k-1}=q[x_1,x_{(k-1)\theta}]$, $s_k=q[x_1,x_{|\dom|}]$. 
Note that the $s_i$ nodes here are not the same as the count for cumulative histogram, so we will use range query $q[x_1,x_i]$ to represent the count $s_i$ in a cumulative histogram.
The first $S$ node, $s_1$ is the root of a subtree consisting of $H$ nodes. This subtree is denoted by $H_1$ and is used for answering all possible range queries within this interval $[x_1,x_{\theta}]$.
For all $1<i\leq k$, $s_i$ has two children: $s_{i-1}$ and the root of a subtree made of $H$ nodes, denoted by $H_i$. Similarly, it also has a fan-out of $f$ and represents counts for values $[(i-1)\theta+1, i\theta]$. We denote the height of the subtree by $h=\ceil{\log_f\theta}$. Using this hybrid structure, we could release cumulative counts in this way: $q[x_1,x_{l\theta}]+q[x_{l\theta+1},x_j]$, where $l\theta\leq j < (l+1)\theta$. Here $q[x_1,x_{l\theta}]$ is answered using $s_l$ and $q[x_{l\theta+1},x_j]$ is answered using $H_l$. Then any range query could be answered as $q[x_i,x_j]=q[x_1,x_j]-q[x_1,x_{i-1}]$.

\stitle{Privacy Budgeting Strategy} 
Given total privacy budget $\epsilon$, we denote the privacy budget assigned to all the $S$ nodes by $\epsilon_S$ and to all the $H$ nodes by $\epsilon_H$. When a tuple change its value from $x$ to $y$, where $d_G(x,y)\leq \theta$ , at most one $S$ node changes its count value and at most $2h$ $H$ nodes change their count values. Hence, for $i=2,...,k$, we add Laplace noise drawn from $Lap(\frac{1}{\epsilon_S})$ to each $s_i$ and we add Laplace noise drawn from $Lap(\frac{2h}{\epsilon_H})$ to each $H$ node in the subtree $H_i$. As $S_1$ is the root of $H_1$, we assign $\epsilon = \epsilon_S + \epsilon_H$ to the tree $H_1$ and hence we add Laplace noise drawn from $Lap(\frac{2h}{\epsilon_H+\epsilon_S})$ to each $H$ node in $H_1$, including $s_1$.
In this way, we could claim that this $OH$ tree satisfies  $(\epsilon, (\dom, G^{d,\theta}, \inp_n))$-Blowfish privacy. 
When $\theta = |\dom|$, $H_1$ is going to be the only tree to have all the privacy budget. This is equivalent to the hierarchical mechanism for differential privacy.
\begin{theorem}
Consider a policy $P=(\dom, G^{d,\theta}, \inp_n)$. (1) The Ordered Hierarchical structure satisfies $(\epsilon, P)$-Blowfish privacy. (2) The expected error of releasing a single count in cumulative histogram or answering a range query this structure over $\dom$ is ,
\begin{equation} \label{eq:error_range}
\error_{q[x_i,x_j],M_{G^{d,\theta}}^{OH}} = O \left( \frac{|\dom |-\theta}{|\dom |\epsilon_S^2} + \frac{(f-1)\log_f^3\theta}{\epsilon_H^2}\right)
\end{equation}
\end{theorem}
\begin{proof} (sketch) 1) Increasing count of $x$ by 1 and decreasing count of $y$ by 1, where $d(x,y)<\theta$ only affect the counts of at most one $S$ node and $2h$ $H$ nodes. Since we draw noise from $Lap(\frac{1}{\epsilon_S})$ for $S$ nodes and from $Lap(\frac{2h}{\epsilon_H})$ for $H$ nodes (where $\epsilon=\epsilon_S+\epsilon_H$), we get $(\epsilon, P)$-Blowfish privacy based on sequential composition.\\
 2) Consider all the counts in cumulative histogram, and there are $|\dom|$ of them, only $|\dom|-\theta$ requires $S$ nodes. Each $S$ node has a error of $\frac{2}{\epsilon_S ^2}$. This gives the first fraction in Eqn. (\ref{eq:error_range}). On average, the number of $H$ nodes used for each count in cumulative histogram is bounded by the height of the $H$ tree and each $H$ node has an error of $\frac{8h^2}{\epsilon_H ^2}$, which explains the second fraction in Eqn. (\ref{eq:error_range}).  
\end{proof}

Each range query $q[x_i,x_j]$ requires at most 2 counts from cumulative histogram and its exact form of expected error is shown below,
\begin{eqnarray} \label{eqn:range_ave_error}
\error_{q[x_i,x_j],M_{G^{d,\theta}}^{OH}} = \frac{c_1}{\epsilon_S^2} +\frac{c_2}{\epsilon_H^2}, 
\end{eqnarray}
where $c_1 = \frac{4(|\dom| -\theta)}{|\dom |+1}$ and $c_2 = \frac{8(f-1)\log_f\theta^3|\dom |}{|\dom |+1}$.
Given $\theta$ and $f$, at $\epsilon_S^* = \frac{c_1^{1/3}}{c_1^{1/3} + c_2^{1/3}}\epsilon $, 
we obtain the minimum 
\begin{eqnarray} \label{eqn:range_ave_error_opt}
 \error_{q[x_i,x_j],M_{G^{d,\theta}}^{OH}}^* = \frac{(c_1^{1/3} + c_2^{1/3})^3}{\epsilon^2}
\end{eqnarray}
In particular, when $\theta=|\dom |$, $c_1 = 0$, this is equivalent to the classical hierarchical mechanism without $S$ nodes and we have $\error_{q[x_i,x_j],M_{G^{d,\theta}}^{OH}} = O(\frac{\log^3|\dom |}{\epsilon^2})$. When $\theta=1$, $c_2 = 0$, this is the pure Ordered Mechanism without using $H$ nodes and $\error_{q[x_i,x_j],M_{G^{d,\theta}}^{OH}}=O(\frac{1}{\epsilon^2})$.

\stitle{Complexity Analysis} The complexity of construction of the hybrid tree $OH$ and answering range query are $O(|\dom |)$ and $O(\log\theta)$ respectively, where $\theta \leq |\dom|$, which is not worse than the classical hierarchical methods.

\newstuff{
\subsection{Empirical Evaluation}\label{sec:cdfeval}
}
We empirically evaluate the error of range queries for the Ordered Hierarchical Mechanism with $(\epsilon, (\dom,G^{d,\theta}, \inp_n))$-Blowfish privacy on two real-world datasets -- \adult and \tweet. The \adult data set\footnote{\url{http://mlr.cs.umass.edu/ml/datasets/Adult}} consists of Census records of $48842$ individuals. We consider the ordinal attribute \caploss with a domain size of $4357$. The \tweet data set is the same dataset used for $k$-means clustering (Sec~\ref{sec:kmeanseval}). Here, in order to have a total ordering for the dataset, we project the \tweet data set on its \lat with a domain size of $400$, around $2222$ km. The fan-out $f$ is set to be 16 and each experiment is repeated 50 times. 
Figure \ref{fig:range_mean_adult} shows the mean square error $\error$ of $10000$ random range queries for various values of $\epsilon=\left\{0.1,0.2,...,0.9,1.0\right\}$. Seven threshold values $\theta =\left\{full, 1000, 500,100,50,10,1\right\}$ are considered. 
For \adult, for example, $\theta = 100$ means the adversary cannot distinguish between values of \caploss within a range of $100$ and $\theta=full$ means the adversary cannot distinguish between all the domain values (same as differential privacy). Figure \ref{fig:range_mean_tweetlat} considers 4 threshold values $\theta=\left\{full, 500km, 50km, 5km \right\}$ for \tweet. 
\newstuff{When $\theta = 1$ (\adult) or $\theta = 5km$ (\tweet), the ordered hierarchical mechanism is same as the ordered mechanism.}
From both figures, we see that as the $\theta$ increases, $\error$ decreases and orders of magnitude difference in  error between $\theta=1$ and $\theta=|\dom|$.

\section{Blowfish with Constraints}
\label{sec:auxiliary}
\newstuff{
In this section, we consider query answering under Blowfish policies with constraints $P = (\dom, G, \inp_Q)$, where $(\inp_Q \subsetneq \inp_n)$. 
In the presence of general deterministic constraints $Q$, pairs of neighboring databases can differ in any number of tuples, depending on structures of $Q$ and the discriminative graph $G$. Computing the policy specific sensitivity in this general case is a hard problem, as shown next.
\begin{theorem}\label{thm:sens_hardness}
Given a function $f$ and a policy \\$P = (\dom, G, \inp_Q)$. Checking whether $S(f,P) > 0$ is NP-hard. The same is true for the complete histogram query $h$.
\end{theorem}
\begin{proof}(sketch)
The proof  follows from the hardness of checking whether 3SAT has at least 2 solutions. The reduction  uses $\inp = \{0,1\}^n$, constraints corresponding to clauses in the formula, and  $\{s^i_0, s^i_1\}_i$ as secret pairs. 
\end{proof}
\cthm\ref{thm:sens_hardness} implies that checking whether $S(f,P) \leq z$ is co-NP-hard for general constraints $Q$. In fact, the hardness result holds even if we just consider the histogram query $h$ and general {\em count query constraints}. 

Hence, in the rest of this section, we will focus on releasing histograms under a large subclass of constraints called {\em sparse count query constraints}. In \csec\ref{sec:auxiliary:sparse} we show that when the count query constraint is ``sparse'', we can efficiently compute  $S(h,P)$, and thus we can use the Laplace mechanism to release the histogram. In \csec\ref{sec:auxiliary:applications}, we will show that our general result about $S(h,P)$ subject to sparse count query constraints can be applied to several important practical scenarios.
}

%
%
%

\subsection{Global Sensitivity for Sparse Constraints}
\label{sec:auxiliary:sparse}




\noindent
%
A {\em count query} $\query_\pred$ returns the number of tuples satisfying predicate $\pred$ in a database $D$, \ie, $\query_\pred(D) = \sum_{t \in D} \one_{\pred(t) = \true}$. The auxiliary knowledge we consider here is a {\em count query constraint} $Q$, which can be expressed as a conjunction of query-answer pairs:
\begin{equation}
\query_{\pred_1}(D) = {\sf cnt}_1 \wedge \query_{\pred_2}(D) = {\sf cnt}_2 \wedge \ldots \wedge \query_{\pred_p}(D) = {\sf cnt}_p.
\end{equation}
Since the answers ${\sf cnt}_1, {\sf cnt}_2, \ldots, {\sf cnt}_p$ do not affect our analysis, we denote the auxiliary knowledge or count query constraint as $Q = \{\query_{\pred_1}, \query_{\pred_2}, \ldots, \query_{\pred_p}\}$. Note that this class of auxiliary knowledge is already very general and commonly seen in practice. For example, marginals of contingency tables, range queries, and degree distributions of graphs can all be expressed in this form. 

Even for this class of constraints, calculating $S(h,P)$ is still hard. In fact, the same hardness result in \cthm\ref{thm:sens_hardness} holds for count query constraints (using a reduction from the {Vertex Cover} problem).

\subsubsection{Sparse Auxiliary Knowledge}

Consider a secret pair $(s^i_x, s^i_y) \in \Spairs^G$ about a tuple $t$ with $t.\_id = i$, and a count query $\query_\pred \in Q$. If the tuple $t \in D$ changes from $x$ to $y$, there are three mutually exclusive cases about $\query_\pred(D)$: 
i) increases by one ($\neg \pred(x) \wedge \pred(y)$), 
ii) decreases by one ($\pred(x) \wedge \neg\pred(y)$), or 
iii) stays the same (otherwise).

\begin{definition}[Lift and Lower] \label{def:affect}
A pair $(x,y)$ $\in$ $\dom \times \dom$ is said to lift a count query $\query_\pred$ iff $\pred(x) = \false \wedge \pred(y) = \true$, or lower $\query_\pred$ iff $\pred(x) = \true \wedge \pred(y) = \false$.
\end{definition}

\noindent Note that one pair may lift or lower many count queries simultaneously. We  now define sparse auxiliary knowledge. 

\begin{definition}[Sparse Knowledge] \label{def:sparseknowledge}
The auxiliary knowledge $Q = \{\query_{\pred_1}, \query_{\pred_2}, \ldots, \query_{\pred_p}\}$ is {\em sparse} w.r.t. the discriminative secret graph $G=(V,E)$, iff each pair $(x,y) \in E$ lifts at most one count query in $Q$ and lowers at most one count query in $Q$. 
\end{definition}

\begin{example}\label{exa:sparsek} {\em (Lift, Lower, and Sparse Knowledge)}
Consider databases from domain $\dom = A_1 \times A_2 \times A_3$, where $A_1 = \{a_1, a_2\}$, $A_2 = \{b_1, b_2\}$, and $A_3 = \{c_1, c_2, c_3\}$ and count query constraint $Q = \{\query_1, \query_2, \query_3, \query_4\}$ as in \cfig\ref{fig:policygraph}(a). With full-domain sensitive information, any pair in $\dom\times\dom$ is a discriminative secret and thus the discriminative secret graph $G$ is a complete graph. A pair $((a_1, b_1, c_1), (a_2, b_2, c_2))$ lifts $\query_4$ and lowers $\query_1$; and a pair $((a_1, b_2, c_1), (a_1, b_2, c_2))$ neither lifts nor lowers a query. We can verify every pair either (i) lifts exactly one query in $Q$ and lowers exactly one in $Q$, or (ii) lifts or lowers no query in $Q$. So $Q$ is sparse w.r.t. the discriminative secret graph $G$.
\end{example}

We will show that when the auxiliary knowledge $Q$ is sparse w.r.t. the discriminative secret graph $G = (V,E)$ in a policy $P = (\dom, G, \inp_Q)$, it is {\em possible} to analytically bound the policy specific global sensitivity $S(h,P)$. To this end, let's first construct a directed graph called {\em policy graph} ${\cal G}_P = ({\cal V}_P, {\cal E}_P)$ from $P$, the count queries in $Q$ forming the vertices, and the relationships between count queries and secret pairs forming edges. 

\begin{figure}[t]
\centering
\includegraphics[scale=0.6]{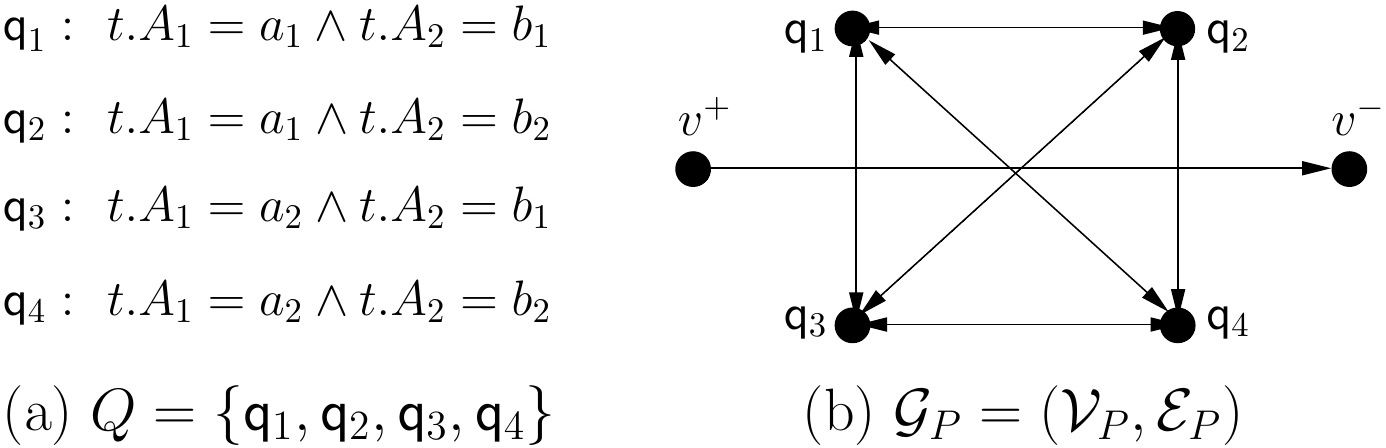}
\caption{Policy graph ${\cal G}_P = ({\cal V}_P, {\cal E}_P)$ of databases with three attributes $A_1 = \{a_1, a_2\}$, $A_2 = \{b_1, b_2\}$, and $A_3 = \{c_1, c_2, c_3\}$ subject to count query constraint $Q = \{\query_1, \query_2, \query_3, \query_4\}$ and full-domain sensitive information}
\label{fig:policygraph}
\end{figure}

\begin{definition}[Policy Graph] \label{def:policygraph} Given a policy $P$ $=$ $(\dom, G(V,E), \inp_Q)$, and a sparse count constraint $Q$, the policy graph ${\cal G}_P = ({\cal V}_P, {\cal E}_P)$ is a directed graph, where
\squishlist
\item ${\cal V}_P = Q \cup \{v^+, v^-\}$: Create a vertex for each count query $\query_\pred \in Q$, and two additional special vertices $v^+$ and $v^-$.
\item ${\cal E}_P$: i) add a directed edge $(\query_\pred, \query_{\pred'})$ iff there exists a secret pair $(x,y) \in E$ lifting $\query_{\pred'}$ and lowering $\query_{\pred}$; ii) add a directed edge $(v^+, \query_\pred)$ iff there is a secret pair in $E$ lifting $\query_\pred$ but not lowering any other $\query_{\pred'}$; iii) add a directed edge $(\query_\pred, v^-)$ iff there is a secret pair in $E$ lowering $\query_\pred$ but not lifting any other $\query_{\pred'}$; and iv) add edge $(v^+, v^-)$.
\squishend
\end{definition}

Let $\alpha({\cal G}_P)$ denote the length (number of edges) of the longest simple cycle in ${\cal G}_P$. $\alpha({\cal G}_P)$ is defined to be $0$ if ${\cal G}_P$ has no directed cycle. Let $\xi({\cal G}_P)$ be the length (number of edges) of a longest simple path from $v^+$ to $v^-$ in ${\cal G}_P$.

\begin{example}\label{exa:policygraph} {\em (Policy Graph)}
Followed by \cexa\ref{exa:sparsek}, since the count query constraint $Q$ is sparse w.r.t. the discriminative secret graph, we have its policy graph ${\cal G}_P = ({\cal V}_P, {\cal E}_P)$ as in \cfig\ref{fig:policygraph}(b). For example, a pair $((a_1, b_1, c_1),$ $(a_2, b_2, c_2))$ in $\dom\times\dom$ lifts $\query_4$ and lowers $\query_1$, so there is an edge $(\query_1, \query_4)$. There is no edge from $v^+$ or to $v^-$, except $(v^+, v^-)$, because every pair in $\dom\times\dom$ either lifts one query and lowers one, or lifts/lowers no query. In this policy graph ${\cal G}_P$, we have $\alpha({\cal G}_P) = 4$ and $\xi({\cal G}_P)=1$.
\end{example}

\begin{theorem}\label{thm:calGS}
Let $h$ be the complete histogram query. In a policy $P = (\dom, G, \inp_Q)$, if the auxiliary knowledge $Q$ is sparse w.r.t. $G$, then we have: 
\[S(h,P) \leq 2 \max\{\alpha({\cal G}_P), \xi({\cal G}_P)\},\]
%
%
If there exist two neighboring databases $(D_1, D_2) \in N(P)$ s.t. $||h(D_1) - h(D_2)||_1 = 2 |T(D_1, D_2)|$ when $|T(D_1, D_2)| = \max_{(D',D'') \in N(P)}|T(D', D'')|$, then we have the equality: \[S(h,P) = 2 \max\{\alpha({\cal G}_P), \xi({\cal G}_P)\}.\]
\end{theorem}
\begin{proof} See Appendix.
\end{proof}

For count query constraint $Q$ that is sparse with respect to policy $P = (\dom, G, \inp_Q)$, we have the following immediate corollary about an upper bound. 
\vspace{-1mm}
\begin{corollary}
In a policy $P = (\dom, G, \inp_Q)$, if $Q$ is sparse w.r.t. $G$, then $S(h, P) \leq 2\max\{|Q|, 1\}$.
\end{corollary}
Thus, drawing noise from Laplace$(2\max\{|Q|, 1\}/\epsilon)$  suffices (but may not be necessary) for releasing the complete histogram while ensuring $(\epsilon, P)$-Blowfish privacy.

%

\subsection{Applications}
\label{sec:auxiliary:applications}
The problem of calculating $\alpha({\cal G}_P)$ and $\xi({\cal G}_P)$ exactly in a general policy graph ${\cal G}_P$ is still a hard problem, but becomes tractable in a number of practical scenarios. We give three such examples: i) the policy specific global sensitivity $S(h,P)$ subject to auxiliary knowledge of one marginal for full-domain sensitive information; ii) $S(h,P)$ subject to auxiliary knowledge of multiple marginals for attribute sensitive information; and iii) $S(h,P)$ subject to auxiliary knowledge of range queries for distance-threshold sensitive information.

\subsubsection{Marginals and Full-domain Secrets}

{\em Marginals} are also called {\em cuboids} in data cubes. Intuitively, in a marginal or a cuboid $C$, we project the database of tuples onto a subset of attributes $[C] \subseteq \{A_1, A_2, \ldots, A_k\}$ and count the number of tuples that have the same values on these attributes. Here, we consider the scenario when the adversaries have auxiliary knowledge about one or more marginals, \ie, the counts in some marginals are known.

\begin{definition}[Marginal]
Given a database $D$ of $n$ tuples from a $k$-dim domain $\dom = A_1 \times A_2 \times \ldots \times A_k$, a $d$-dim marginal $C$ is the (exact) answer to the query:

~~{\sf SELECT $A_{i_1}$, $A_{i_2}$, \ldots, $A_{i_d}$, COUNT($*$) FROM $D$}

~~{\sf GROUP BY $A_{i_1}$, $A_{i_2}$, \ldots, $A_{i_d}$}

\noindent Let $[C]$ denote the set of $d$ attributes $\{A_{i_1}, A_{i_2}, \ldots, A_{i_d}\}$.
%
\end{definition}

A marginal $[C] = \{A_{i_1}, A_{i_2}, \ldots, A_{i_d}\}$ is essentially a set of count queries $C^\query = \{\query_\pred\}$, where the predicate $\pred(t) := (t.A_{i_1} = a_{i_1}) \wedge (t.A_{i_2} = a_{i_2}) \wedge \ldots \wedge (t.A_{i_2} = a_{i_d})$, for all possible $(a_{i_1}, a_{i_2}, \ldots, a_{i_d}) \in A_{i_1} \times A_{i_2} \times \ldots \times A_{i_d}$.

Let ${\cal A} = \{A_1, A_2, \ldots, A_k\}$ be the set of all attributes. For a marginal $C$, define $\size(C) = \prod_{A_i \in [C]} |A_i|$, where $|A_i|$ is the cardinality of an attribute $A_i$. So $\size(C)$ is the number of possible rows in the marginal $C$, or the number of the count queries in $C^\query$ constructed as above.
%

Suppose a marginal with $[C] \subsetneq {\cal A}$ is known to the adversary. Let $\inp_{Q(C)}$ denote the set of databases with marginal $C$ equal to certain value. Recall that we want to publish the complete histogram $h$ of a database $D$ from a domain $\dom = A_1 \times A_2 \times \ldots \times A_k$. For the full-domain sensitive information, using \cthm\ref{thm:calGS}, we have the global sensitivity equal to $2~\size(C)$.

%

\begin{theorem}
Let $h$ be the complete histogram. For a policy $P = (\dom, G, \inp_{Q(C)})$, where $G$ represents the full-domain sensitive information $\Spairs^{\rm full}$ and $[C] \subsetneq {\cal A}$ is a marginal, we have  $S(h,P) = 2~\size(C)$.
\end{theorem}
\begin{proof} (sketch)
Consider the set of count queries $C^\query$. It is not hard to show that $C^\query$ is sparse w.r.t. the complete graph $G$. So we can construct a policy graph from $(\dom, G, \inp_{C^\query})$, which is a complete graph with vertex set $C^\query$. From \cthm\ref{thm:calGS}, we have $S(f,P) = 2|C^\query| = 2~\size(C)$. The upper bound $S(h,P) \leq 2|C^\query|$ is directly from \cthm\ref{thm:calGS}, and it is not hard to construct two neighboring databases to match this upper bound as $[C] \subsetneq{\cal A}$.
\end{proof}

\begin{example}
Continuing with \cexa\ref{exa:policygraph}, note that the constraints in \cfig\ref{fig:policygraph}(a) correspond to the marginal $[C] = \{A_1, A_2\}$. So from (i) in \cthm\ref{thm:calGS}, we have $S(h,P) \leq 2 \times 4 = 8$. The worst case $S(h,P) = 8$ can be verified by considering the two neighboring databases $D_1$ and $D_2$, each with four rows: $a_1b_1c_1$ (in $D_1$)/$a_1b_2c_2$ (in $D_2$), $a_1b_2c_1$/$a_2b_1c_2$, $a_2b_1c_1$/$a_2b_2c_2$, and $a_2b_2c_1$/$a_1b_1c_2$.
\end{example}

\subsubsection{Marginals and Attribute Secrets}

Now suppose a set of $p$ marginals $C_1,$ $\ldots,$ $C_p$ with $[C_1],$ $\ldots,$ $[C_p]$ $\subsetneq \{A_1, A_2, \ldots, A_k\}$ are auxiliary knowledge to the adversary. Let $\inp_{Q(C_1, \ldots, C_p)}$ be the set of databases with these $p$ marginals equal to certain values. For the attribute sensitive information, if the $p$ marginals are disjoint, using \cthm\ref{thm:calGS} the global sensitivity is $2\max_{1 \leq i \leq p} \size(C_i)$.

\begin{theorem}
Let $h$ be the complete histogram. Consider a policy $P = (\dom, G^{\rm attr}, \inp_{Q(C_1, \ldots, C_p)})$, where $[C_i] \subsetneq {\cal A}$ for any marginal $C_i$, and $[C_i] \cap [C_j] = \emptyset$ for any two $C_i$ and $C_j$. Then we have $S(h,P) = 2\max_{1 \leq i \leq p} \size(C_i)$.
\end{theorem}
\begin{proof} (sketch)
Consider the set of count queries $Q = C_1^\query \cup \ldots \cup C_p^\query$, it is not hard to show that $Q$ is sparse w.r.t. $G^{\rm attr}$. The policy graph from $(\dom, G^{\rm attr}, Q)$ is the union of $p$ cliques with vertex sets $C_1^\query$, \ldots, $C_p^\query$. From \cthm\ref{thm:calGS}, we have $S(h,P) = 2\max_i|C_i^\query| = 2\max_{i} \size(C_i)$. The upper bound $S(h,P) \leq 2\max_{i}|C_i^\query|$ is directly from \cthm\ref{thm:calGS}, and it is not hard to construct two neighboring databases to match this upper bound as $[C_i] \neq  {\cal A}$.
\end{proof}

\subsubsection{Grid and Distance-threshold Secrets}

Our general theorem about $S(f,P)$ can be also applied to databases with geographical information.

Consider a domain $\dom = [m]^k$, where $[m] = \{1,2, \ldots, m\}$. When $k = 2$ or $3$, $\dom$ can be used to approximately encode a $2$-dim plane or a $3$-dim space. For two points $x, y \in \dom$, we define {\em distance} $d(x,y)$ to be the $L^p$ distance $||x-y||_p$. For two point sets $X, Y \subset \dom$, we define $d(X,Y) = \min_{x \in X, y \in Y} d(x,y)$. A geographical database $D$ consists of $n$ points, each of which is drawn from the domain $\dom$ and may represent the location of an object.

Define a rectangle $R = [l_1, u_1] \times [l_2, u_2] \times \ldots \times [l_k, u_k]$, where $l_i \in [m]$, $u_i \in [m]$, and $l_i \leq u_i$. A {\rm range count query} $\query_R$ returns the number of tuples whose locations fall into the rectangle $R$. 
$R$ is called a {\em point query} if $l_i = u_i$ for all $i$. 

In this scenario, suppose the answers to a set of $p$ range count queries are known to the adversary. So we can represent the auxiliary knowledge as $Q = \{\query_{R_1},$ $\query_{R_2},$ $\ldots,$ $\query_{R_p}\}$. Also, suppose we aim to protect the distance-threshold sensitive information $\Spairs^{d, \theta} \!\! =$ $\{(s^i_x, s^i_y)$ $\mid$ $d(x,y)$ $\leq$ $\theta\}$ while the publishing complete histogram $h$.

Using \cthm\ref{thm:calGS}, we can calculate the global sensitivity if all rectangles are disjoint, \ie, $R_i \cap R_j = \emptyset$ for any $i \neq j$, as follows. Construct a graph ${\cal G}_R(Q) = ({\cal V}_R, {\cal E}_R)$ on the set of rectangles in $Q$: i) create a vertex in ${\cal V}_R$ for rectangle $R_i$ in each range count query $\query_{R_i}$ in $Q$; and ii) add an edge $(R_i, R_j)$ into ${\cal E}_R$ iff $d(R_i, R_j) \leq \theta$. 
We can prove that the policy specific global sensitivity equals to $2(\maxcomp(Q)+1)$ when there are no point query constraints, where $\maxcomp(Q)$ is the number of nodes in the largest connected component in ${\cal G}_R(Q)$. Note that $\maxcomp(Q)$ (and hence $S(h,P)$) can be computed efficiently. 

\begin{theorem}
Let $h$ be the complete histogram. For a policy $P = (\dom, G, \inp_{Q})$, where $\dom = [m]^k$, $G$ represents the distance-threshold sensitive information $\Spairs^{d, \theta}$ ($\theta > 0$), and $Q$ is a set of disjoint range count queries $\{\query_{R_1},$ $\query_{R_2},$ $\ldots,$ $\query_{R_p}\}$ with $R_i \cap R_j = \emptyset$ for $i \neq j$. We have $S(h,P) \leq 2(\maxcomp(Q)+1)$. If none of the constraints are point queries, then $S(h,P) = 2(\maxcomp(Q)+1)$.
\end{theorem}
%




\balance 

\section{Conclusions}
\label{sec:conclusions}
We propose a new class of privacy definitions, called Blowfish privacy, with the goal of seeking better trade-off between privacy and utility. The key feature of Blowfish is a  policy, where users can specify sensitive information that needs to be protected and knowledge about their databases which has been released to potential adversaries. Such a rich set of ``tuning knobs'' in the policy enable users to improve the utility by customizing sensitive information and to limit attacks from adversaries with auxiliary knowledge. Using examples of kmeans clustering, cumulative histograms and range queries, we show how to tune utility using reasonable policies with weaker specifications of privacy. For the latter, we develop strategies that are more accurate than any differentially private mechanism. 
Moreover, we study how to calibrate noise for Blowfish policies with count constraints when publishing histograms, and the general result we obtain can be applied in several practical scenarios.

\vspace{5pt}
\noindent
\textbf{Acknowledgements}: We would like to thank Jiangwei Pan and the anonymous reviewers for their comments. This work was supported by the National Science Foundation under Grant \# 1253327 and a gift from Google.


\appendix
\section{Proof of Theorem 7.1}\label{proof_calGS}
\begin{proof}(sketch)
Recall that the policy specific global sensitivity is defined as
\[
S(h,P) = \max_{(D_1, D_2) \in N(P)} ||h(D_1) - h(D_2)||_1.
\]

\smallskip\noindent {\bf Direction I ($S(h,P) \leq 2 \max\{\alpha({\cal G}_P), \xi({\cal G}_P)\}$).} It suffices to prove that for any two databases $D_1, D_2 \in \inp_Q$, if $|T(D_1, D_2)| > \max\{\alpha({\cal G}_P), \xi({\cal G}_P)\}$, there must exist another database $D_3 \in \inp_Q$ s.t. $T(D_1, D_3) \subsetneq T(D_1, D_2)$, \ie, $(D_1, D_2) \notin N(P)$; and thus for any two databases $(D_1, D_2) \in N(P)$, we have $||h(D_1) - h(D_2)||_1 \leq 2|T(D_1, D_2)| \leq$ $2\max\{\alpha({\cal G}_P),$ $\xi({\cal G}_P)\}$ which implies $S(h,P) \leq 2\max\{\alpha({\cal G}_P),$ $\xi({\cal G}_P)\}$.

To complete the proof, we consider two databases $D_1, D_2 \in \inp_Q$ with $|T(D_1, D_2)| > \max\{\alpha({\cal G}_P), \xi({\cal G}_P)\}$, and show how to construct the $D_3$ defined above.

First of all, for any secret pair $(s^i_x, s^i_y) \in T(D_1, D_2)$, it must lift and/or lower some count query $\query_\pred \in Q$; otherwise, we can construct $D_3$ by changing the value of tuple $t$ with $t.\_id = i$ in $D_1$ into its value in $D_2$.

To construct $D_3$, now let's consider a directed graph ${\cal G}_{D_1|D_2} = ({\cal V}_{D_1|D_2}, {\cal E}_{D_1|D_2})$, where ${\cal V}_{D_1|D_2} \subseteq {\cal V}_P$ and ${\cal E}_{D_1|D_2}$ is a multi-subset of ${\cal E}_P$ (\ie, an edge in ${\cal E}_P$ may appear multiple times in ${\cal E}_{D_1|D_2}$). ${\cal E}_{D_1|D_2}$ is constructed as follows: for each $(s^i_x, s^i_y) \in T(D_1,D_2)$, i) if $(x, y)$ lifts $\query_{\pred'}$ and lowers $\query_{\pred}$, add a directed edge $(\query_\pred, \query_{\pred'})$ into ${\cal E}_{D_1|D_2}$; ii) if $(x, y)$ lifts $\query_\pred$ but not lowering any other $\query_{\pred'}$, add an edge $(v^+, \query_\pred)$; and iii) if $(x, y)$ lowers $\query_\pred$ but not lifting any other $\query_{\pred'}$, add an edge $(\query_\pred, v^-)$. ${\cal V}_{D_1|D_2}$ is the set of count queries involved in ${\cal E}_{D_1|D_2}$.

${\cal G}_{D_1|D_2}$ is Eulerian, \ie, each vertex has the same in-degree as out-degree except $v^+$ and $v^-$ (if existing in ${\cal G}_{D_1|D_2}$), because of the above construction and the fact that $D_1, D_2 \in \inp_Q$. As $|{\cal E}_{D_1|D_2}| = |T(D_1, D_2)| > \max\{\alpha({\cal G}_P), \xi({\cal G}_P)\}$ (\ie, ${\cal G}_{D_1|D_2}$ is larger than any simple cycle or simple $v^+$-$v^-$ path in ${\cal G}_P$) and ${\cal G}_{D_1|D_2}$ is Eulerian, ${\cal G}_{D_1|D_2}$ must have a proper subgraph which is either a simple cycle or a simple $v^+$-$v^-$ path. Let ${\cal E}_{D_1 \rightarrow D_2}$ be the edge set of this simple cycle/path. Construct $D_3$ that is identical to $D_1$, except that for each secret pair $(s^i_x, s^i_y)$ associated with each edge in ${\cal E}_{D_1 \rightarrow D_2}$, the value of tuple $t$ with $t.\_id = i$ is changed from $x$ to $y$. We can show that $D_3$ satisfies its definition, and thus the proof for Direction I is completed.

\smallskip\noindent {\bf Direction II ($S(h,P) \geq 2\max\{\alpha({\cal G}_P), \xi({\cal G}_P)\}$).} Let's first prove a weaker inequality:
\begin{equation} \label{equ:GS:directionII}
\max_{(D_1,D_2) \in N(P)}|T(D_1, D_2)| \geq \max\{\alpha({\cal G}_P), \xi({\cal G}_P)\}.
\end{equation}
It implies $S(h,P) \geq 2\max\{\alpha({\cal G}_P), \xi({\cal G}_P)\}$ if the condition in (ii) of the theorem holds. Combined with Direction I, we can conclude $S(h,P) = 2\max\{\alpha({\cal G}_P), \xi({\cal G}_P)\}$.

To prove \eqref{equ:GS:directionII}, it suffices to show that for any simple cycle/$v^+$-$v^-$ path in ${\cal G}_P$, we can construct two databases $D_1$ and $D_2$ s.t. $(D_1, D_2) \in N(P)$ and $|T(D_1, D_2)| = \hbox{its length}$. Consider a simple cycle $\query_{\pred_1},$ $\query_{\pred_2},$ $\ldots,$ $\query_{\pred_l},$ $\query_{\pred_{l+1}}=\query_{\pred_1}$. Starting with any database $D \in \inp_Q$, let $D_1 \leftarrow D$ and $D_2 \leftarrow D$ initially. For each edge $(\query_{\pred_i}, \query_{\pred_{i+1}})$, from the definition of policy graphs, we can find a secret pair $(x,y) \in E(G)$ s.t. $\left(\neg\query_{\pred_i}(x) \wedge \query_{\pred_i}(y)\right) \wedge \left(\query_{\pred_{i+1}}(x) \wedge \neg\query_{\pred_{i+1}}(y)\right)$; create two new tuples: $t_1.{\_id} = t_2.{\_id} = i$, $t_1 = x$, and $t_2 = y$; and then let $D_1 \leftarrow D_1 \cup \{t_1\}$ and $D_2 \leftarrow D_2 \cup \{t_2\}$. It is not hard to verify that finally we get two databases $D_1$ and $D_2$ s.t. $(D_1, D_2) \in N(P)$ and $|T(D_1, D_2)| = \hbox{cycle length}$. The proof is similar for a simple $v^+$-$v^-$ path. 
\end{proof}

\section{Proof of Theorem 4.1}\label{proof_seqComp}
\begin{proof}(sketch)
Let $M_{M_1,M_2}$ denote the mechanism that outputting the results of $M_1$ and $M_2$ sequentially.
As $M_1$ satisfies $(\epsilon_1,P)$-Blowfish privacy, for every pair of neighboring databases $(D_a,D_b)\in N(P)$,
and every result $r_1\in range(M_1)$, we have  
\begin{eqnarray}
 Pr[M_1(D_a)=r_1] \leq e^{\epsilon_1}Pr[M_1(D_b)=r_1]
\end{eqnarray}
The result of $M_1$ is outputted before the result of $M_2$,
so $r_1$ will turn out to be another input of $M_2$, together with the original dataset.
As $M_1$ satisfies $(\epsilon_1,P)$-Blowfish privacy,
for every pair of neighboring databases $(D_a,D_b)\in N(P)$ coupling with the same $r_1$, 
and for every result $r_2\in range(M_2)$, we have
\begin{eqnarray}
 Pr[M_2(D_a,r_1)=r_2] \leq e^{\epsilon_1}Pr[M_2(D_b,r_1)=r_2]
\end{eqnarray}
Therefore, for every pair of neighboring databases $(D_a,D_b)\in N(P)$,
and every set of output sequence $(r_1,r_2)$, we have  
\begin{eqnarray}
&& Pr[M_{M_1,M_2}(D_a)=(r_1,r_2)] \nonumber \\
&=& Pr[M_1(D_a)=r_1]Pr[M_2(D_a,r_1)=r_2] \nonumber \\
&\leq& e^{\epsilon_1}Pr[M_1(D_b)=r_1]e^{\epsilon_2}Pr[M_2(D_b,r_1)=r_2] \nonumber \\
&\leq& e^{\epsilon_1+\epsilon_2}Pr[M_1(D_b)=r_1]Pr[M_2(D_b,r_1)=r_2] \nonumber \\
&=& e^{\epsilon_1+\epsilon_2} Pr[M_{M_1,M_2}(D_b)=(r_1,r_2)]
\end{eqnarray}
\end{proof}

\section{Proof of Theorem 4.2-4.3}\label{proof_parComp}
\begin{proof}(sketch)
For every pair of neighboring databases $(D_a,D_b)\in N(P)$ with the cardinality constraint or with disjoint subsets of constraints $Q_1,...,Q_p$, there is only one subset  of \_ids, let's say $S_i*$, with different values in $D_a$ and $D_b$ while $D_a\cap S_i = D_b \cap S_i$ for all $i\neq i^*$. Hence,  for every set of output sequence $r$,
\begin{eqnarray}
&& Pr[M(D_a)=r] = \prod_i Pr[M_i(D_{a} \cap S_i)=r_i] \nonumber \\
&\leq &  e^{\epsilon_{i^*}} Pr[M_{i^*}(D_{b}\cap S_{i^*})=r_{i^*}] \prod_{i,i\neq i^*}Pr[M_i(D_{b}\cap S_i)=r_i] \nonumber \\
&\leq & e^{\max_i \epsilon_i} \prod_i Pr[M_i(D_{b}\cap S_i)=r_i] \nonumber \\
&=& e^{\max_i \epsilon_i}Pr[M(D_b)=r]
\end{eqnarray}
\end{proof}

\end{document}